\documentclass[10pt,letter]{article}

\usepackage{amsmath} 
\usepackage{amsfonts} 
\usepackage{amssymb}
\usepackage{amsthm}
\usepackage{eucal}
\usepackage{mathrsfs}
\usepackage{enumerate}
\usepackage{graphicx} 

\usepackage{alltt}
\usepackage{setspace}
\usepackage{times}

\usepackage{array}
\newcolumntype{S}{>{\centering\arraybackslash} m{.4\linewidth} }

\usepackage{natbib}

\bibpunct{(}{)}{;}{a}{,}{,}

\newcommand{\R}{\mathbb{R}}

\newcommand{\E}{\mathbb{E}}
\newcommand{\I}{\mathbb{I}}

\newcommand{\balpha}{\boldsymbol{\alpha}}
\newcommand{\bbeta}{\boldsymbol{\beta}}
\newcommand{\bvarphi}{\boldsymbol{\varphi}}
\newcommand{\btau}{\boldsymbol{\tau}}

\newcommand{\bsigma}{\boldsymbol{\sigma}}
\newcommand{\bxi}{\boldsymbol{\xi}}
\newcommand{\beps}{\boldsymbol{\epsilon}}
\newcommand{\bphi}{\boldsymbol{\phi}}

\renewcommand{\vec}[1]{\mathbf{#1}}
\newcommand{\mat}[1]{\mathbf{#1}}

\newcommand{\abs}[1]{\left\vert#1\right\vert}
\newcommand{\norm}[1]{\left\Vert#1\right\Vert}

\newtheorem{thm}{Theorem}[section]

\newtheorem{lemma}[thm]{Lemma}
\newtheorem{prop}[thm]{Proposition}
\theoremstyle{definition}

\newtheorem{example}[thm]{Example}
\theoremstyle{remark}

\begin{document}
\begin{titlepage}
\begin{center}
{\Large\bf Information Flow in Interaction Networks}
\end{center}
\vspace{.35cm}

\begin{center}
{\large Aleksandar Stojmirovi\'c\, and Yi-Kuo Yu\footnote{to whom correspondence should be addressed}}
\vspace{0.25cm}
\small

\par \vskip .2in \noindent
National Center for Biotechnology Information\\
National Library of Medicine\\
National Institutes of Health\\
Bethesda, MD 20894\\
United States
\end{center}

\normalsize
\vspace{0.25cm}

\begin{abstract}
Interaction networks, consisting of agents linked by their interactions, are ubiquitous accross many disciplines of modern science. Many methods of analysis of interaction networks have been proposed, mainly concentrating on node degree distribution or aiming to discover clusters of agents that are very strongly connected between themselves. These methods are principally based on graph-theory or machine learning.

We present a mathematically simple formalism for modelling context-specific information propagation in interaction networks based on random walks. The context is provided by selection of sources and destinations of information and by use of potential functions that direct the flow towards the destinations. We also use the concept of dissipation to model the aging of information as it diffuses from its source. 

Using examples from yeast protein-protein interaction networks and some of the histone acetyltransferases involved in control of transcription, we demonstrate the utility of the concepts and the mathematical constructs introduced in this paper.
\end{abstract}
\end{titlepage}

\section{Introduction}

Interaction networks are abundant and have recently gained significant publicity in many diverse modern disciplines such as electronics \citep{CJS01}, sociology \citep{WF94,Newman04} and epidemiology \citep{BBPV05}. In its simplest form, an interaction network consists of a collection of entities (or agents), where two agents  are linked if they interact in some way. For example, in an acquaintance network \citep{ASBS00}, the agents represent persons and two persons are linked together if they know each other while the Woldwide Web network consists of web pages with links between pages \citep{BKFP00}. Mathematically, networks correspond exactly to graphs (or multigraphs), with agents as vertices and links as edges, which can be weighted and/or directed depending on the exact application being modeled. The key to analysis of interaction networks is the assumption of information transitivity: information can flow through or can be exchanged via paths of interactions.

Biology in post-genomic era also contains numerous examples of molecular networks \citep{Galitski04}. Metabolic networks have been modeled by representing metabolites as nodes and chemical reactions as links: two metabolites are linked if they participate in the same reaction \citep{MZ03}. Genetic networks have genes as nodes with two genes being linked if they interact through directed transcriptional regulation \citep{GBBK02}. Protein-protein interaction networks have proteins as nodes, with the links representing physical interactions (binding) between proteins \citep{PHJ04}. Large scale high-throughput studies in model organisms such as \textit{Saccharomyces cerevisiae} (baker's yeast) \citep{ICOYHS01,UGCM00}, \textit{Drosophilla melanogaster} (fruit-fly) \citep{GBBC03}, \textit{Caenorhabditis elegans} (roundworm) \citep{LABG04} and humans \citep{SWLH05,RVHH05}, provided extensive datasets of protein-protein interactions, stored in publicly-available databases such as the Database of Interacting Proteins (DIP) \citep{XSDHKE02,SMSPBE04}. Unfortunately, there is very little consistency between the protein-protein interaction data coming from different high-throughput experiments \citep{SSM03} and significant effort has been expended in devising ways to discover false positives and false negatives \citep{SSRS06}. This problem is not restricted to protein-protein interactions: microarray data also contains non-negligible inconsistencies \citep{MM04} .

Numerous approaches have been proposed for analysis of biological and, in particular, protein-protein interaction networks \citep{AS06}. However, due to space restrictions, we will refer to just a few. Most algorithms aim to discover `functional modules' \citep{HHLM99}, representing well connected clusters of nodes with the same or similar function, by using clustering techniques from graph theory and/or machine learning \citep{SPAD02,SM03,RG03,PEO04,NJAC05,XHDZ05,CSW06,CY06,HCZR06}. Very frequently, these techniques make use of additional experimental data which is not present in the network structure itself. For example, methods for discovery of complexes from protein-protein interaction networks often refer to the data from dataset from different species \citep{KSKS03,SIKS05,SSKK05}, microarray expression studies \citep{SPAD02,CY06}, or human-curated functional classifications \citep{NJAC05,CSW06}.

Our approach to analyzing interaction networks is very different, relying solely on the network structure. We model diffusion of information through the network by discrete-time random walks moving from the nodes representing the sources of information to their destinations. The choice of sources and destinations provides the \emph{context of analysis} with the nodes most affected by information flow being called \emph{Information Transduction Modules}. We use two modes of diffusion, dual to each other, which we call absorbing and emitting, with our absorbing mode directly corresponding to deeply investigated absorbing Markov chains \citep{KS76}. Random walks and corresponding Markov chains are one of the subjects of spectral graph theory \citep{Chung97} but we do not use eigenspace decomposition in our work, instead relying on a basic matrix algebra approach similar to that of \citet{KS76}. 

The algorithm \textit{Functional Flow} by \citet{NJAC05}, also modeling diffusion of information from sources, is closest to our emitting model. However, to delineate a certain biological context, we additionally direct the flow from sources to selected destinations using potential functions and allow the information content to dissipate (evaporate) from the network at each time step, thus modeling natural `aging' of information.

Our models allow investigation of several types of biological questions from protein-protein interaction networks. Many proteins perform their function in cooperation with other proteins through, often large, protein complexes. Thus, to elucidate the function of a given protein, it is useful to know the most likely members of complexes it may belong to and their relations to each other. Additionally, if two proteins are known to have similar function, what, if any, are the proteins they share in their respective complexes? To help answer such questions, we employ our absorbing diffusion mode. 

The answers to the above questions can provide the general interaction environment of one or more proteins. It is also very instructive to identify specific modules mediating interactions between distant (in network terms) proteins. Our emitting diffusion mode can be used to find possible candidates for members of such modules. Furthermore, analysis of interaction modules obtained from considering different proteins in the same biological context may lead to discovery of fundamental units of information transduction. To achieve this we developed the concept of information interference. More concrete definitions will be presented in the body of the text.

This paper is organized as follows. Section 2 outlines the theory behind our models of information diffusion in networks. For better readability, all the theorems and proofs, using mainly the basic concepts and results from the matrix algebra are given in Appendix (the reader may wish to consult the standard linear algebra textbooks such as \citep{HK71} or \citep{BR97} for background). Section 3 introduces the methods of analysis of results obtained using the concepts of Section 2, while Section 4 presents concrete examples centered around yeast histone acetyltransferases. We finish with discussion and conclusion in Section 5.

\section{Theory}

\subsection{Preliminaries}

We represent an interaction network as a weighted directed graph $\Gamma=(V,E,w)$ where $V$ is a finite set of vertices of size $n$, $E\subseteq V\times V$ is a set of edges and $w$ is a non-negative real-valued function on $V\times V$ that is positive on $E$, giving the weight of each edge (the weight of non-existing edge is defined to be $0$). Assuming an ordering of vertices in $V$, we represent a real-valued function on $V$ as a state (column) vector $\vec{\bvarphi}\in\R^n$ and the connectivity of $\Gamma$ by the \emph{weight} matrix $\mat{W}$ where $W_{ij}=w(i,j)$ (the weight of an edge from $i$ to $j$). If $\Gamma$ is an unweighted undirected graph, $\mat{W}$ is the adjacency matrix of $\Gamma$ where
\begin{equation}
W_{ij} = \begin{cases} 2 &\text{if $i=j$ and $(i,i)\in E$},\\ 1 &\text{if $i\neq j$ and $(i,j)\in E$}, \\ 0 &\text{if $(i,j)\not\in E$.}\end{cases}
\end{equation}
Throughout this paper, we will not make distinction between a vertex $v\in V$ and its corresponding state given by a particular ordering of vertices.

Let  $\mat{P}$ denote the $n\times n$ \emph{transition} matrix of $\Gamma$ where 
\begin{equation}
P_{ij} = \frac{W_{ij}}{\sum_k W_{ik}},
\end{equation}
that is, $\mat{P}$ is the weight matrix of $\Gamma$ normalized by row. The matrix $\mat{P}$ can be used to model random walks on $\Gamma$: for any pair of vertices $i$ and $j$, $P_{ij}$ gives the probability of the random walk moving from vertex $i$ to vertex $j$ in one time step, which is proportional to the weight $W_{ij}$. Since the matrix $\mat{P}$ is stochastic (all rows sum to unity), it can also be interpreted as the transition matrix for Markov chain on the set $V$. In the following sections we will model information diffusion as a random walk on $\Gamma$ with particular starting and terminating points.

\subsection{Constrained diffusion}\label{subsec:transition}

In this section we select certain vertices as sources or sinks of information and solve for the number of times a vertex is visited. Let $S$ denote the set of selected vertices, let $T=V\setminus S$ and let $m=\abs{T}$. Assuming that the first $n-m$ states correspond to vertices in $S$, we write the matrix $\mat{P}$ in the canonical form:
\begin{equation}\label{eqn:Pcannonical}
\mat{P}=\left[ \begin{array}{cc}\mat{P}_{SS} & \mat{P}_{ST}\\ \mat{P}_{TS} & \mat{P}_{TT}\end{array}\right].
\end{equation}
Here $\mat{P}_{AB}$ denotes a matrix giving probabilities of moving from $A$ to $B$ where $A,B$ stand for either $S$ or $T$. The states (vertices) belonging to the set $T$ are called \emph{transient}.

\subsubsection{Absorption in sinks}

Suppose now that the set $S$ represents the set of \emph{sinks} of information: any information reaching a sink vertex is absorbed and cannot not leave it. Let $\mat{F}(t)$ denote an $m \times (n-m)$ matrix such that $F_{ij}(t)$ is the probability that the information originating at $i\in T$ is absorbed at $j\in S$ in $t$ or fewer steps. Since information can only be absorbed once in any state $s\in S$, it follows that the information reaching $j$ avoided all other sinks. For the same reason, $F_{ij}(t)$ can be interpreted as the expected number of visits to the state $j$ of a random walk starting at $i$ for all times up to $t$.

Absorption at $j$ after not more than $t$ steps can be achieved in two ways: either the content reached vertex $j$ in the first step, with probability $P_{ij}$ or it moved to some transient vertex $k$ in the first step and was absorbed by $j$ from there in at most $t-1$ steps, with probability $P_{ik}F_{kj}(t-1)$. Therefore, we have for all $t=1,2,\ldots$,
\begin{equation}\label{eqn:sink0}
F_{ij}(t+1) = P_{ij} + \sum_{k\in T} P_{ik}F_{kj}(t),
\end{equation}
or in the matrix form
\begin{equation}\label{eqn:sink1}
\mat{F}(t+1) = \mat{P}_{TS} + \mat{P}_{TT}\mat{F}(t).
\end{equation}

We solve for the long-term or equilibrium state, where $\mat{F}(t+1)=\mat{F}(t)=\mat{F}$. In this case, Equation (\ref{eqn:sink1}) becomes
\begin{equation}\label{eqn:sink2}
\mat{F} = \mat{P}_{TS} + \mat{P}_{TT}\mat{F},
\end{equation}
or 
\begin{equation}\label{eqn:sink3}
(\I-\mat{P}_{TT})\mat{F} = \mat{P}_{TS},
\end{equation}
where $\I$ denotes the identity matrix. If $\I-\mat{P}_{TT}$ is invertible, let $\mat{G}=(\I-\mat{P}_{TT})^{-1}$. Equation (\ref{eqn:sink3}) then has a unique solution
\begin{equation}\label{eqn:sink4}
\mat{F} = \mat{G}\mat{P}_{TS}.
\end{equation}

\subsubsection{Diffusion from sources}

Now consider the dual problem where $S$ is a set of sources of information. Each source emits a unit of information at each time step and no information can enter any source: we assume any information entering a source vanishes. Let $\mat{H}(t)$ denote an $(n-m) \times m$ matrix such that $H_{ij}(t)$ is the total expected number of times the transient vertex $j$ is visited by a random walk emitted from source $i$ for the time up to $t$.

The information emitted from $i$ can arrive at $j$ at time $t$ in two different ways: either the content was emitted from $i$ at time $t$ and reached $j$ directly, or it was emitted at an earlier time step, was located at some transient vertex at time $t-1$ and moved from there to $j$ at time $t$. The former option contributes $P_{ij}$ while the latter contributes $H_{ik}(t-1)P_{kj}$ for all $k\in T$ towards $H_{ij}$. Therefore, we have for all $t=1,2,\ldots$,
\begin{equation}\label{eqn:source0}
H_{ij}(t+1) = P_{ij} + \sum_{k\in T} H_{ik}(t)P_{kj},
\end{equation}
or in the matrix form
\begin{equation}\label{eqn:source1}
\mat{H}(t+1) = \mat{P}_{ST} + \mat{H}(t)\mat{P}_{TT}.
\end{equation}
Similarly to the previous case, we are interested in the steady state, representing the total expected number of visits, where $\mat{H}(t+1)=\mat{H}(t)=\mat{H}$. In this case, Equation (\ref{eqn:source1}) becomes
\begin{equation}\label{eqn:source2}
\mat{H} = \mat{P}_{ST} + \mat{H}\mat{P}_{TT},
\end{equation}
or 
\begin{equation}\label{eqn:source3}
\mat{H}(\I-\mat{P}_{TT}) = \mat{P}_{ST}.
\end{equation}
If $\I-\mat{P}_{TT}$ is invertible, Equation (\ref{eqn:source3}) has a unique solution
\begin{equation}\label{eqn:source4}
\mat{H} = \mat{P}_{ST}\mat{G}.
\end{equation}

\subsubsection{Existence and interpretation of solutions}\label{subsubsec:Green}

It can immediately be observed that existence of solutions to Equation (\ref{eqn:source3}) and Equation (\ref{eqn:sink3}) are equivalent: they both depend on the existence of the inverse of $\I-\mat{P}_{TT}$. Specifically, they are special cases of the discrete Laplace equation on $T$ with the Dirichlet boundary condition on $S$ \citep{Chung97,CY00}.

Given a square matrix $\mat{M}$, the matrix $\I - \mat{M}$ is often called the \emph{discrete Laplace operator} of $\mat{M}$. Let $\Delta=\I - \mat{P}_{TT}$ ($\Delta$ is the discrete Laplace operator of $\mat{P}$ restricted to $T$). Equation (\ref{eqn:sink3}) can then be written as
\begin{equation}\label{eqn:laplace1}
\Delta\mat{F} = \mat{P}_{TS}.
\end{equation}
Denote by $\vec{e}_k$ the $k$-th standard basis (column) vector of length $n-m$ where $(\vec{e}_k)_j=\delta_{kj}$ ($\delta$ here is the Kronecker's delta). Let $\vec{f}_k=\mat{F}\vec{e}_k$ denote the $k$-th column of $\mat{F}$ and let $\vec{p}_k=\mat{P}_{TS}\vec{e}_k$. Then, solving Equation (\ref{eqn:laplace1}) is equivalent to solving the discrete Laplace equation
\begin{equation}\label{eqn:laplace2}
\Delta\vec{f}_k = \vec{p}_k
\end{equation}
for all $k\in S$. The standard basis vectors $\vec{e}_k$ provide exactly the \emph{Dirichlet boundary conditions} on the set $S$ (the set $S$ can be assumed to be a boundary of $T$).

It is also easy to see that Equation (\ref{eqn:source3}) can be written as 
\begin{equation}\label{eqn:laplace3}
\mat{H}\Delta = \mat{P}_{ST}.
\end{equation}
Hence, the solution to (\ref{eqn:laplace3}) is obtained by solving the discrete Laplace equation in terms of the discrete Laplace operator of the transpose of $\mat{P}$. 

The \emph{Green's function} is defined to be the inverse of the Laplacian. In our case the inverse of $\Delta$ is exactly the matrix $\mat{G}=(\I-\mat{P}_{TT})^{-1}$ and hence the existence of solutions to Equations (\ref{eqn:source3}) and (\ref{eqn:sink3}) is equivalent to existence of the Green's functions to the corresponding Laplacian. In the absorbing Markov chain theory \citep{KS76}, the matrix $\mat{G}$ is known as the \emph{Fundamental matrix} of the corresponding absorbing Markov chain. The entry $G_{ij}$ represents the mean number of times the random walk reaches vertex $j\in T$ having started in state $i\in T$. 

We now present some elementary sufficient conditions for existence of the Green's functions of the discrete Laplacians of the graphs. The full proofs are given in Appendix \ref{app:Green}. For the development of the discrete Green's functions (for undirected graphs) in terms of the eigenvalues and eigenfunctions of the Laplacian, we refer the reader to the paper by \citet{CY00}.

\begin{prop}\label{prop:specrad0}
Suppose that $\Gamma$ is a weighted directed graph such that for every $p\in T$ there exists $s\in S$ such that there exists a directed path from $p$ to $s$. Then, the matrix $\I-\mat{P}_{TT}$ is invertible and 
\begin{equation}\label{eq:invassum0}
(\I-\mat{P}_{TT})^{-1} = \sum_{k=0}^\infty (\mat{P}_{TT})^k.
\end{equation}
\end{prop}

Proposition \ref{prop:specrad0} thus guarantees existence of the Green's functions if every transient vertex can be connected to a source or sink via a directed path. If the underlying graph is undirected, this condition can be rephrased as follows: every connected component of $V$ contains at least one vertex from $S$. 

In the context of information diffusion, the connectivity condition implies that all information entering the transient set at any specific time must eventually leave it, either by absorption into $S$ when $S$ is a set sinks, or by dissipation when $S$ represents the set of sources. We will further discuss the concept of dissipation in \ref{sec:dissipation}.

Assuming the Green's function exists, the entries of the matrices $\mat{F}$ and $\mat{H}$ can be interpreted in several different ways. Fundamentally, both $F_{ij}$ and $H_{ij}$ represent the total expected number of times the vertex $j$ is visited by the information originating at the vertex $i$ while avoiding all members of the boundary set $S$ (the proofs are given in Appendix \ref{app:interpr1}). It is also clear, by Equation (\ref{eq:invassum0}), that $\mat{F}$ and $\mat{H}$ are both non-negative matrices and that $\mat{F}=\lim_{t\to\infty} \mat{F}(t)$ and $\mat{H}=\lim_{t\to\infty} \mat{H}(t)$. In addition, the rows of $\mat{F}$ all sum to $1$ (Lemma \ref{lemma:Fstoch} in Appendix \ref{app:interpr2}) and thus $F_{ij}$ is the overall probability an information originating from transient vertex $i$ is absorbed at the sink $j$ while avoiding all other sinks. 

If we assume that a random walk deposits a fixed amount of information content each time it visits a node, we can interpret $H_{ij}$ is the overall amount of information content originating from the source $i$ deposited at the transient vertex $j$. If $\Gamma$ is an undirected graph with symmetric weight matrix $\mat{W}$ and $S$ contains a single source, the value of $H_{ij}$ is directly proportional to the degree of the transient vertex $j$ (Appendix \ref{app:interpr2}). Hence, in this case, the total average number of times of visits for each transient node is proportional to its degree. This is no longer true if $\mat{W}$ is not symmetric.

Furthermore, we can interpret $F_{ij}$ as the sum of probabilities of paths originating at the vertex $i\in T$ and terminating at the vertex $j\in S$ that avoid all other nodes in the set $S$, and $H_{ij}$ as the sum of probabilities of paths originating at the vertex $i\in S$ and terminating at the vertex $j\in T$, also avoiding all other nodes in the set $S$. Each such path has a finite but unbounded length. However, unlike $F_{ij}$, $H_{ij}$ does not represent a probability because the events of the information being located at $j$ at the times $t$ and $t'$ are not mutually exclusive (a random walk can be at $j$ at time $t$ and revisit it at time $t'$). For $F_{ij}$, the absorbing events at different times are mutually exclusive.

\subsection{Information dissipation}\label{sec:dissipation}

It was mentioned previously that the requirement that every transient node is connected to a node in the set $S$ is effectively equivalent to the property that all information content entering the transient set leaves it at the nodes in $S$. In the present section we extend our model to allow the information to dissipate not only at those nodes but also at the transient nodes.

Let $\balpha$ and $\bbeta$ be vectors of length $n$ such that for all $i\in V$, $\alpha_i>0$ and $\beta_i>0$. We form the matrix 
$\mat{\tilde{P}}$ with entries
\begin{equation}
\tilde{P}_{ij} = \alpha_i\beta_jP_{ij},
\end{equation}
and use the new matrix to compute the matrices $\mat{\tilde{F}}$ and $\mat{\tilde{H}}$ by replacing the matrix $\mat{P}$ in the previous section with $\tilde{P}$ so that.
\begin{equation}\label{eq:dissipate1}
\mat{\tilde{F}} = \mat{\tilde{G}}\mat{\tilde{P}}_{TS}.
\end{equation}
and 
\begin{equation}\label{eq:dissipate2}
\mat{\tilde{H}} = \mat{\tilde{P}}_{ST}\mat{\tilde{G}}.
\end{equation}
where $\mat{\tilde{G}}=(\I-\mat{\tilde{P}}_{TT})^{-1}$, provided $\I-\mat{\tilde{P}}_{TT}$ is invertible.

The entry $\alpha_i$ gives the proportion of the signal leaving the vertex $i$ that is retained (we call the value of $1-\alpha_i$ the \emph{outgoing dissipation coefficient} of the node $i$) while the entry $\beta_j$ gives the proportion of the signal entering the vertex $j$ that is retained  (the value $1-\beta_j$ is called the \emph{incoming dissipation coefficient} of the node $j$). The case where $\alpha_i=\beta_i=1$ for all $i\in V$ gives back the original matrix $\mat{P}$. Note that our definition allows entries of $\balpha$ and $\bbeta$ that are greater than $1$, corresponding to negative dissipation coefficients. Such coefficients lead to amplification of the signal. However, in order for the Green's function $\mat{\tilde{G}}$ to exist, any amplification should be balanced by dissipation. 

We now establish a sufficient condition for existence of $\mat{\tilde{G}}$. The proof, as well as a discussion of its generalization, is given in Appendix \ref{app:dissipation}.

\begin{prop}\label{prop:dissipation}
Let $\alpha_*=\max\{\alpha_i:i\in V\}$ and $\beta_*=\max\{\beta_i:i\in V\}$ and suppose $\alpha_*\beta_*<1$. Then, the matrix $\I-\mat{\tilde{P}}_{TT}$ is invertible and 
\begin{equation}\label{eq:invassum1}
(\I-\mat{\tilde{P}}_{TT})^{-1} = \sum_{k=0}^\infty (\mat{\tilde{P}}_{TT})^k.
\end{equation}
\end{prop}

Proposition \ref{prop:dissipation} makes no assumptions on the connectivity of the graph: the equilibrium solutions exist regardless of the graph topology. The reason for the removal of the connectivity conditions is that a unit of information originating anywhere in the network has a nonzero probability of being dissipated at each time step and therefore will disappear in the long term, with a portion possibly reaching a sink in the absorbing model. The vectors of coefficients $\alpha$ and $\beta$ provide us with the ability to consider different rates of dissipation at different vertices. We demonstrate the utility of the extended model in examples involving protein-protein interaction networks (Section \ref{sec:examples}), where we use vertex specific dissipation to construct `evaporating nodes' that dissipate most of the information coming in but allow unrestricted outward flow. 

A possible further generalization of this model is for the entries of the vectors $\balpha$ and $\bbeta$ to be functions of the state variable of the dynamical system instead of constants. The dynamical system in this case would become non-linear, allowing us to model amplification or dissipation of the information depending on the time specific state of the system.

\subsection{Potentials}

Our models so far, including the dissipation modifications described above, model `free diffusion' of information through the network: the likelihood for the signal to move from vertex $i$ to vertex $j$ is proportional to the relative weight of the edge $(i,j)$ among all edges emanating from $i$ (dissipation only affects the total amount transmitted). In order to direct the flow of information towards or away from selected nodes, we adjust the weights of edges of our network graph $\Gamma$ using \emph{potentials}, real-valued monotone functions defined on the nodes that depend on the distances from selected points.
 
Let $\rho$ denote the path-metric on the weighted directed connected graph $\Gamma=(V,E,w)$, where for all $i,j\in V$, $\rho(i,j)$ denotes the sum of the reciprocals of the weights of the edges forming the shortest directed path from $i$ to $j$. Suppose 
$R$ is a subset of $T$ such that for each $k\in R$ there exists a monotone potential function  $\theta_k:\R\to\R$. For each vertex $j\in V$ define the \emph{total potential} at $j$, denoted $\Theta(j)$ by
\begin{equation}
\Theta(j) = \sum_{k\in R} \theta_k(\rho(j,k)).
\end{equation}
Let $\hat{\Gamma}$ denote the new weighted directed graph $(V,E,\hat{w})$ where
\begin{equation}\label{eq:pseudosinkpt}
\hat{W}_{ij} = W_{ij}\exp\left(- \Theta(j)\right).
\end{equation}
The form of Equation (\ref{eq:pseudosinkpt}) ensures that the signal preferentially diffuses from each vertex towards the vertices adjacent to it that have lower potential relative to other adjacent vertices. 

A vertex $i\in V$ is called a \emph{destination} if $\Theta$ has a minimum at $i$. There can be multiple destinations in a network. The natural candidates for destinations are the members of the set $S$ since all information entering them does not leave them. Some transient states, with the weights of their outgoing edges adjusted to partially accumulate the signal, are also good candidates for destinations. 

Let $K$ be a subset of $T$ and let $0\leq\gamma\leq 1$. From the already modified graph $\hat{\Gamma}$, we form the graph $\Gamma'$ represented by the weight matrix $\mat{W}'$ where
\begin{equation}\label{eq:Rgenmat1}
W'_{ij} = \begin{cases} \hat{W}_{ij} & \text{if $i\not\in K$,}\\
\gamma \hat{W}_{ij} &  \text{if $i\in K$ and $i\neq j$,}\\
\hat{W}_{ij} + (1-\gamma)\sum_{k\neq i} \hat{W}_{ik} & \text{if $j\in K$ and $i = j$}.\end{cases}
\end{equation}
The effect of this modification is to turn each vertex $i\in K$, called a \emph{pseudosink}, into a partial sink: some proportion of the weights of edges emanating out of $i$ is transferred to the edge pointing back to $i$. The parameter $\gamma$, representing the proportion of information allowed to leave each pseudosink while the remainder is accumulated, is called the \emph{pseudosink leakage coefficient}. The value 
$\gamma=1$ implies no change in edge weights.

The value $\gamma=0$ is a special case because no directed path exists between pseudosinks and source nodes in the resulting graph $\Gamma'$ and Proposition \ref{prop:specrad0} does not apply. In this case, there are two possibilities leading to the existence of the Green's function: either set the outgoing dissipation coefficient of the pseudosinks to something less than $1$, or treat the pseudosinks as parts of the boundary set $S$, as a `non-emitting source' defined in \ref{subsec:emittingmodel} below.

Note that, while dissipation is applied to the transition matrix $\mat{P}$, potentials and pseudosinks are applied to the weight matrix $\mat{W}$ prior to normalization. Since applications of potentials and pseudosinks do not commute, potentials are applied before pseudosinks, although pseudosinks can be potential centers (members of the set $R$).

\section{Theoretical Methods for Analysis}

In the previous section we introduced the basic concepts related to our models of diffusion of information through networks as well as some modifications to the underlying graph and the transition matrix that lead to biologically realistic models. After all modifications are applied, we obtain the matrices $\mat{\tilde{F}}$ and $\mat{\tilde{H}}$, the Green's functions arising where $S$ represents sinks and sources, respectively. Here we turn to the practical interpretation of these results, which depend on the boundary conditions imposed on the vertices in $S$.

\subsection{Absorbing model}

In the case where $S$ represents sinks of information (the \emph{absorbing model}), the entries of the matrix $\mat{\tilde{F}}$ have a clear probabilistic interpretation: $\tilde{F}_{ij}$ is the probability that information starting at transient vertex $i$ reaches the sink $j$ while avoiding all other sinks, taking into account the dissipation as well as the new weights induced by the potentials. Generally, each sink $j$ exerts a `region of influence', including the transient points with large $\tilde{F}_{ij}$. Depending on the distributions of sinks within the network, some transient node may have a
$\tilde{F}_{ij}$ small for all $j$: information emerging from these points is more likely to dissipate than to reach any of the sinks.

If $S'\subset S$ is a selection of sink nodes, then $\sum_{j\in S'} F_{ij}$ gives the total probability of information reaching the set $S'$ from the vertex $i$, avoiding all other nodes in $S$. In this context, we call the nodes in $S'$ \emph{explicit sinks} (since we investigate the probabilities of reaching them) and the remaining nodes in $S$ \emph{implicit sinks}, the points that serve as sinks of information but are not considered. Furthermore, if the sinks are treated as general boundary points, with boundary values not restricted to $0$ and $1$, the entries of $\mat{\tilde{F}}$ can be interpreted as temperatures \citep{ZBY07}.

\subsection{Emitting model}\label{subsec:emittingmodel}

Where $S$ represents sources (the \emph{emitting model}), the entries of $\mat{\tilde{H}}$ can be interpreted as visiting times or as information contents: $\tilde{H}_{ij}$ is the total information content emitted from the source $i$ deposited at the transient vertex $j$. Information is dissipated at all sources and the value of $\tilde{H}_{ij}$ is dependent on transient dissipation coefficients $\balpha$ and $\bbeta$ and the potentials. For biological applications, we will consider the case where at least one pseudosink is present in addition to one or several sources, with the potential directing the flow towards the pseudosinks. The distribution of entries of the $i$-th row of $\mat{\tilde{H}}$ will then describe the \emph{information transduction module} (ITM) involved in transfer of information from $i$ to the pseudosinks, with the nodes with largest entries being most significant.

Let $\bxi$ denote the vector of length $\abs{S}$ such that for all $i\in S$, $\xi_i\geq 0$. We call $\xi_i$ the \emph{source strength} of the source $i$, representing the amount of information emitted from $i$ at each time step. In this context, we call $i\in S$ an \emph{emitting source} if $\xi_i>0$ and a \emph{non-emitting source} if $\xi_i=0$. Non-emitting sources are essentially information `black holes', dissipating any information coming in and not emitting any.

\subsubsection{Total content}

For any $i\in S$, let $\beps_i$ denote the standard $i$-th row basis vector of length $n-m$, where $(\beps_i)_j=\delta_{ij}$. For $x>0$ define the vector $\bphi_i$ by
\begin{equation}
\bphi_i = \xi_i\beps\mat{\tilde{H}},
\end{equation}
that is, $\bphi_i$ denotes the $i$-th row of $\mat{\tilde{H}}$ multiplied by $\xi_i$. Its entries give the amount of information content originating from the source $i$ of strength $\xi_i$ deposited at transient vertices. The value of $\norm{\bphi_i}_1$ is then the total amount of content originating at source $i$ deposited at the transient states. In our examples in the following sections we choose the source strengths $\bxi$ so that $\norm{\bphi_i}_1$ is the same for all $i\in S$ (we call the resulting vectors $\phi_i$ normalized content vectors). The \emph{joint information content} vector, denoted $\btau$, is defined by
\begin{equation}
\btau = \sum_{i\in S} \bphi_i.
\end{equation}
The vector $\btau$ implicitly depends on the matrix $\mat{\tilde{H}}$ and the source strength vector $\bxi$: we have $\btau=\bxi\mat{\tilde{H}}$.

\subsubsection{Participation ratio}

Let $\vec{x}\in\R^n$ be any vector and recall that for any $0\leq p<\infty$, the $\ell_p$-norm of $\vec{x}$, denoted $\norm{\vec{x}}_p$, is given by $\norm{\vec{x}}_p = \left(\sum_k \abs{x_k}^p \right)^{1/p}$. Define the \emph{participation ratio} of $\vec{x}$, denoted $\pi(\vec{x})$ by 
\begin{equation}
\pi(\vec{x}) = \frac{\norm{\vec{x}}_1^2}{\norm{\vec{x}}_2^2} = \frac{\left(\sum_{k} \abs{x_k} \right)^2}{\sum_k x_k ^2}.
\end{equation}
Participation ratio is well known under a slightly different definition in the physics literature \citep{Thouless74}. It gives the number of components of $\vec{x}$ whose magnitude is `significant'. Clearly, $\pi$ is independent of the scale of $\vec{x}$: we have for any $\lambda>0$, $\pi(\lambda\vec{x})=\pi(\vec{x})$. We illustrate the usage by examples.

\begin{example}
Let $\vec{x}=[1,1,1,1,1]$. Then, $\pi(\vec{x})=\frac{5^2}{5}=5$. All components are equally significant and this is reflected in the participation ratio.
\end{example}
\begin{example}
Now consider $\vec{x}=[1,1,0,0,0]$. We have, $\pi(\vec{x})=\frac{2^2}{2}=2$. Only the first two components are non-zero and are of equal magnitude. 
\end{example}
\begin{example}
Finally, let $\vec{x}=\left[1,\frac12,\frac14,\frac18,\frac{1}{16}\right]$. We obtain $\pi(\vec{x})\approx 2.8181$. Here all five components are non-zero but their magnitudes differ significantly. The participation ratio here implies that the first two components and to a large extent the third are significant while the remaining two are much smaller. 
\end{example}

In our biological examples, we use $\pi(\btau)$ to choose the number of the transient vertices with largest total mass to display as a `significant' subgraph, together with all sources and pseudosinks.

\subsubsection{Interference}

Given the vector of source strengths $\bxi$, the entry of $\tau_j$ can be interpreted as providing the total amount of information deposited at the vertex $j$. It is also possible to investigate the interaction of the signals from different sources using the concept of destructive interference.

For any vector $\vec{x}\in\R^n$, let $\mu$ denote an \emph{interference function} such that $0\leq \mu(\vec{x})\leq \norm{\vec{x}}_1$. When applied to a vector containing information content from different sources, interference function is interpreted as removing some of the information present due to the interaction of the various information types and returning the remaining information content. Interference functions can take various forms depending on the nature of the types of information in each application.

\begin{example}
Suppose $\vec{x}$ consists of two components representing information types that are assumed to completely cancel out each other. In this case, the interference function takes the form $\mu(\vec{x}) = \abs{x_1-x_2}$.
\end{example}
\begin{example}
When $\vec{x}$ has more than two components, there are may possible ways to generalize the above example. We distinguish two general modes of interference: exclusive and partial. Exclusive interference mode represents the case where simultaneous presence of all types of information is necessary for destructive interference. For example, if each information type carries the same weight, the interference function is:
\begin{equation}\label{eq:interf1}
\mu(\vec{x}) =  \sum_{k} \left(x_k - \nu\right),
\end{equation}
where $\displaystyle\nu = \min_{k} x_k$.
\end{example}

\begin{example}
We call the partial interference the case where presence of all types of information is not necessary. It can be modeled in many ways depending on the desired interpretation. For example, if there are three sources, we can use complex numbers to set $\mu$ so that
 \begin{equation}
\mu(\vec{x}) = \abs{\sum_{k=1}^3  x_k\exp\left(\frac{\iota k\pi}{3} \right)}, 
\end{equation}
where $\iota$ denotes the imaginary unit. In this case, some content is lost when any two types of signal are present but all three must be present for complete annihilation.
\end{example}

Given the interference function $\mu$, define the \emph{interference strength function} $\psi:\R^n\to\R\cup\{\infty\}$ by
\begin{equation}\label{eq:interfstr1}
\psi(\vec{x}) = \begin{cases} \norm{\vec{x}}_1 \log\left(\frac{\norm{\vec{x}}_1}{\mu(\vec{x})}\right) & \text{if $\norm{\vec{x}}_1>0$,}\\
0 & \text{if $\norm{\vec{x}}_1=0$.}
\end{cases}
\end{equation}
By the definition of $\mu$

Since $0\leq \mu(\vec{x})\leq \norm{\vec{x}}_1$, it follows that $\psi$ takes non-negative values (including $+\infty$). The value of $\psi$ is infinite if $\mu(\vec{x})=0$ (perfect interference) and finite otherwise. For an $m\times n$ matrix $\mat{X}$ define the vector $\bsigma(\mat{X})$ of length $n$ having the components
\begin{equation}\label{eq:interfstr2}
\sigma_i(\mat{X}) = \psi(\mat{X}\vec{e}_i)
\end{equation}
(recall that $\vec{e}_i$ is the standard column basis vector and hence $\mat{X}\vec{e}_i$ represents the $i$-th column of $\mat{X}$). We will call $\bsigma$ the \emph{interference strength vector}.

For our applications, the entries of the matrix $\mat{X}$ above are interpreted as information contents over some graph: $X_{ij}$ is the the content of type $i$ at the vertex $j$. For each node $j$, the $\ell_1$-norm in Equation (\ref{eq:interfstr1}) can be interpreted in this context as the total information content at $j$ and the value of $\mu$ applied to the $j$-th column of $\mat{X}$ as the information content remaining after interference. Hence, interference strength of each node measures how much information content was lost by interference, adjusted by the node's joint information content.

The matrix $\mat{\tilde{H}}$ is therefore a natural input to $\psi$ and $\bsigma$, however other derived matrices can be used such as $\mat{\tilde{H}}$ adjusted for source strength by multiplying each row by its corresponding source strength $\xi_i$. Furthermore, rows of $\mat{X}$ can come from different $\mat{\tilde{H}}$ matrices, using different potentials or dissipation coefficients, as long as the underlying vertex set is the same. The general purpose of interference strength is to measure the amount of interaction or overlap between different ITMs.

\section{Biological Examples}
\label{sec:methods}
\label{sec:examples}

The theory and methods outlined in previous sections can be applied to any interaction network. This section will present some examples using biological networks, more specifically, yeast protein-protein interaction networks. Since the interaction data obtained using many high-throughput methods is generally inconsistent \citep{SSM03}, we use the core yeast dataset from DIP, version ScereCR20060402, consisting of 2554 proteins and 5952 interactions for all our examples. The core dataset, obtained using the methods of  \citet{DSXE02}, contains only the most reliable interactions from the DIP dataset of all yeast protein-protein interactions. 

Our examples are restricted to investigation of information transduction  modules related to yeast histone acetyltransferases (HATs). Histones are nuclear proteins that are major components of eukaryotic chromatin \citep{Wolffe92}: eukaryotic DNA is organized as a repeating array of nucleosomes consisting of 146 bp of DNA wound around a histone octamer consisting of two of each of histone proteins H2A (Hta1, Hta2 in yeast), H2B (Htb1, Htb2 in yeast), H3 (Hht1, Hht2 in yeast) and H4 (Hhf1, Hhf2 in yeast). It has been repeatedly demonstrated that transcription is strongly influenced by the chromatin structure and DNA-histone interactions in particular. The regions of DNA that interact with histones are generally unavailable for transcription and transcriptional activation and deactivation are connected with chromatin alterations \citep{Wolffe01}.

Histone acetyltransferases are enzymes that acetylate histones, leading to weakening of the nucleosome structure and making the DNA involved accessible to transcription factors \citep{Struhl98,WK98}. \textit{Saccharomyces cerevisiae} contains several HATs from two major classes with a variety of biological functions and substrate specificities \citep{SB00}. The proteins Hat1, Gcn5, Elp3, Spt10 and Hpa2 belong to the GNAT superfamily \citep{NL97}, while Esa1, Sas2 and Sas3 belong to the MYST family \citep{BSABBCCDDFHMVWH96,SEGSPZCLA98}. The proteins TAF1 (TATA-binding protein associated factor), a subunit of the TFIID complex, and Nut1 (Med5), a subunit of the mediator complex \citep{BY05}, have also been associated with histone acetyltransferase activity \citep{MYKBBOWWBKNA96,LBGMK00}.

Unfortunately, the core dataset does not contain the relevant data for all known HATs. The HATs Hpa2 and Spt10 are not present in the core while HAT1 has interactions only with Hat2 and its substrate Hhf2. We chose to primarily concentrate on HATs Gcn5, Esa1 and Elp3 because they are well researched and the interaction data is abundant. They are all involved in transcriptional activation, unlike Sas2, which promotes silencing \citep{OSMBYSW01}. 

Gcn5 is the best characterized of all HATs, preferentially acetylating histone H3 \citep{SS99}. It forms the catalytic subunit of the ADA and SAGA transcriptional activation complexes \citep{GDCRBCOOAWBW97}. In addition to Gcn5, the SAGA complex also contains the proteins Tra1, TAF5, TAF6, TAF9, TAF10, TAF12, Hfi1 (Ada1), Ada2, Ngg1 (Ada3), Spt3, Spt7, Spt8 and Spt20 (Ada5) \citep{TT05}. The ADA complex contains a subset of proteins from the SAGA complex, namely Gcn5, Hfi1, Ada2, Ngg1 and Spt20, plus the adaptor protein Ahc1 \citep{ESSHYBW99}. The TAF proteins in SAGA also belong to the TFIID complex, which overall consists of 15 subunits including a TATA-binding protein and 14 TAFs \citep{SW00}. 

Esa1 is the catalytic subunit of the NuA4 histone acetyltransferase complex essential for growth in yeast \citep{SEGSPZCLA98,AUSCGBPWC99} that catalyses acetlyaltion of the histone H4. It has been established that the NuA4 complex, containing, in addition to Esa1, the proteins Tra1, Epl1 Yng2, Eaf1, Eaf2, Eaf3, Eaf5, Eaf6, Act1, Arp4 and Yaf9, is recruited by a variety of transcriptional complexes as a transcriptional coactivator and is involved in DNA repair \citep{DC04}. 

Elp3 is a part of the six component elongator complex , which is associated with RNA polymerase II during transcript elongation \citep{WOBFEOLATS99}. The elongator complex also includes the proteins Iki3 (Elp1), Elp2--4, Iki1 (Elp5) and Elp6 \citep{KG01}. 

This section contains four examples of the application of our models, depicted in Figures \ref{fig:fig1}--\ref{fig:emitting2b}. Subsection \ref{subsec:absorbing} describes possible complexes associated with the HATs Gcn5, Esa1 and Elp3, taken individually and in competition, that can be inferred from the protein-protein interaction network using the absorbing model. Subsection \ref{subsec:emitting} investigates possible physical interaction interfaces between the MADS box protein Mcm1 \citep{SS95} and the HATs Esa1 and Gcn5. In this case, the emitting model is employed to discover the pathways through which Mcm1 can recruit the above HATs and whether they are recruited through the same interface. Before presenting our results we describe the model parameters and computational techniques used.

\subsection{Parameters and computation}

\subsubsection{Dissipation}

For all our examples, we set $\alpha_i=1$ for every node $i$ in our interaction network so that the outgoing flow from any node is not dissipated. Modeling the incoming dissipation the coefficients $\beta_i$ can take two values: one for `ordinary' and one for \emph{evaporating} vertices. In our examples that use the absorbing model (\ref{subsec:absorbing}), $\beta_i$ is set to $0.70$ for ordinary nodes and $0.01$ for evaporating nodes while the examples using the emitting model (\ref{subsec:emitting}) set 
$0.87$ for ordinary nodes and $0.01$ for evaporating nodes. The evaporating nodes consisted of cytoskeleton proteins Act1, Myo1, Myo2, Myo3, Myo4, Myo5, Smy1, Smy2, Sla1, Arc40, Arp2, Rvs167, Tpm1, Tpm2, Aip1 and Las17 and histones (Hta1, Hta2, Htb1, Htb2, Hht1, Hhf2, Htz1, Hho1).

The coefficients for the ordinary nodes were chosen using the following reasoning. For the emitting model we considered the dissipation rate that would allow the random walk emitted from the source to reach an `average' node along the shortest path to it with the probability slightly less than $0.5$, say $0.49$. We found that the average length of the shortest path between two points in the yeast core dataset is $5.23$ and hence our coefficient is $0.49^{(1/5.23)}=0.872$, which is rounded to $0.87$. A different coefficient was needed for the absorbing examples because we were interested in only the immediate complexes containing our selected HATs: the coefficient $\beta_i=0.87$ would lead to most of the members of the RNA polymerase II holoenzyme to be retrieved as members of the resulting ITM. We chose to consider the shortest paths of length $2$, rather than of the average length $5.23$. Using the same calculation as above, we obtain $0.49^{(1/2)}=0.7$.

The reason for having evaporating nodes with larger dissipation rate is that both the cytoskeleton proteins and the histones form extended structures in the cell and the nucleus, respectively. In our physical interaction network, we assume that information can flow from one protein to another through an intermediate node if all three nodes are brought close together in space and time. Information is not likely to flow through proteins that are parts of extended structures because proteins with completely different biological function may bind them at different locations and at different times. Therefore, 
allowing significant information flow through such nodes would yield biologically implausible results. 

However, depending on the exact context of the investigation, such nodes may have an important role to play and removing them completely from the interaction networks or assigning them to the boundary set $S$ would not be appropriate. Hence, we set a very high incoming dissipation rate at evaporating nodes while allowing the information to originate from them. In terms of our models, this approach means that the evaporating nodes will have very small visiting times in the emitting models and hence will not be components of any ITM. On the other hand, depending on the exact network topology, they may be part of ITMs obtained by the emitting model. Note that other proteins that bind their interacting partners in a non space and time specific manner can be chosen as additional evaporating nodes; we chose histones and cytoskeleton proteins due to their direct relevance to our selected examples.

\subsubsection{Potentials}

All our examples use attracting potentials centered at each pseudosink or sink. The potential function, heuristic in nature, is the same in every example has the the form
\begin{equation}\label{eq:pseudosinkpt2}
\theta_k(x) = \begin{cases} a_1x & \text{if $0<x\leq b$,}\\  a_1x +  a_2(x-b)^2 & \text{if $x>b$,}  \end{cases}
\end{equation}
where $a_1=0.8181$, $a_2= 0.05$, $b=2$ and $k$ is any pseudosink or a sink. The potential function shown above is long-range, affecting the whole graph, with a linear portion for short ranges $0\leq x\leq 2$ and quadratic for distances larger than $2$. We do not expect to see qualitative changes in the results if the form of the potential function is modified as long as it has the effect of attracting information towards the destination.

The sources (in the case of emitting models) and evaporating points were excluded from the graph prior to calculating distances (their distances from the centers were set to an arbitrary large number) in order to exclude the paths passing through them from consideration. The reason for excluding the paths passing through sources was that, by construction, the information never enters a source from a transient vertex, while the evaporating points were excluded because most of the signal entering them is dissipated. 

\subsubsection{Numerical implementation}

The code for computation of the results was implemented in the Python programming language, using the NumPy and SciPy packages \citep{JEP01}. In particular, the computation of the matrices $\mat{\tilde{F}}$ and $\mat{\tilde{G}}$ (Equations (\ref{eq:dissipate1}--\ref{eq:dissipate2})) was performed by the embedded FORTRAN code from the UMFPACK \citep{Davis04} solver of sparse systems of linear equations, using the Automatically Tuned Linear Algebra Software (ATLAS) \citep{WP05} implementation of Basic Linear Algebra Subprograms (BLAS). The graphical representations of the subgraphs of interest were produced by the \textit{neato} program from the Graphviz graph visualization suite \citep{GN00}.

\subsection{HAT complexes: absorbing examples}\label{subsec:absorbing}

\begin{figure}
\begin{center}
\begin{tabular}[t]{lr}
\textbf{(a)} & \\
& \includegraphics[scale=0.70]{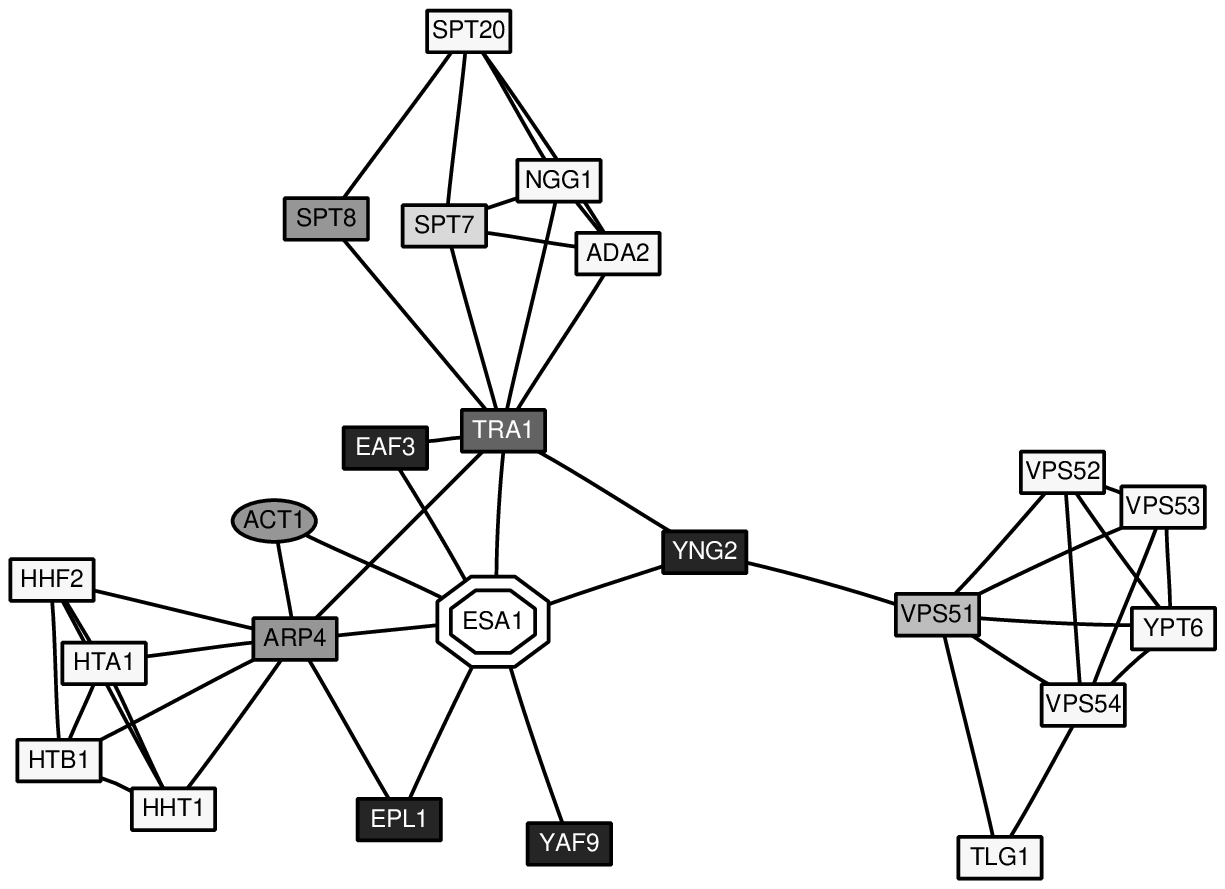} \\
\textbf{(b)} & \\
& \includegraphics[scale=0.70]{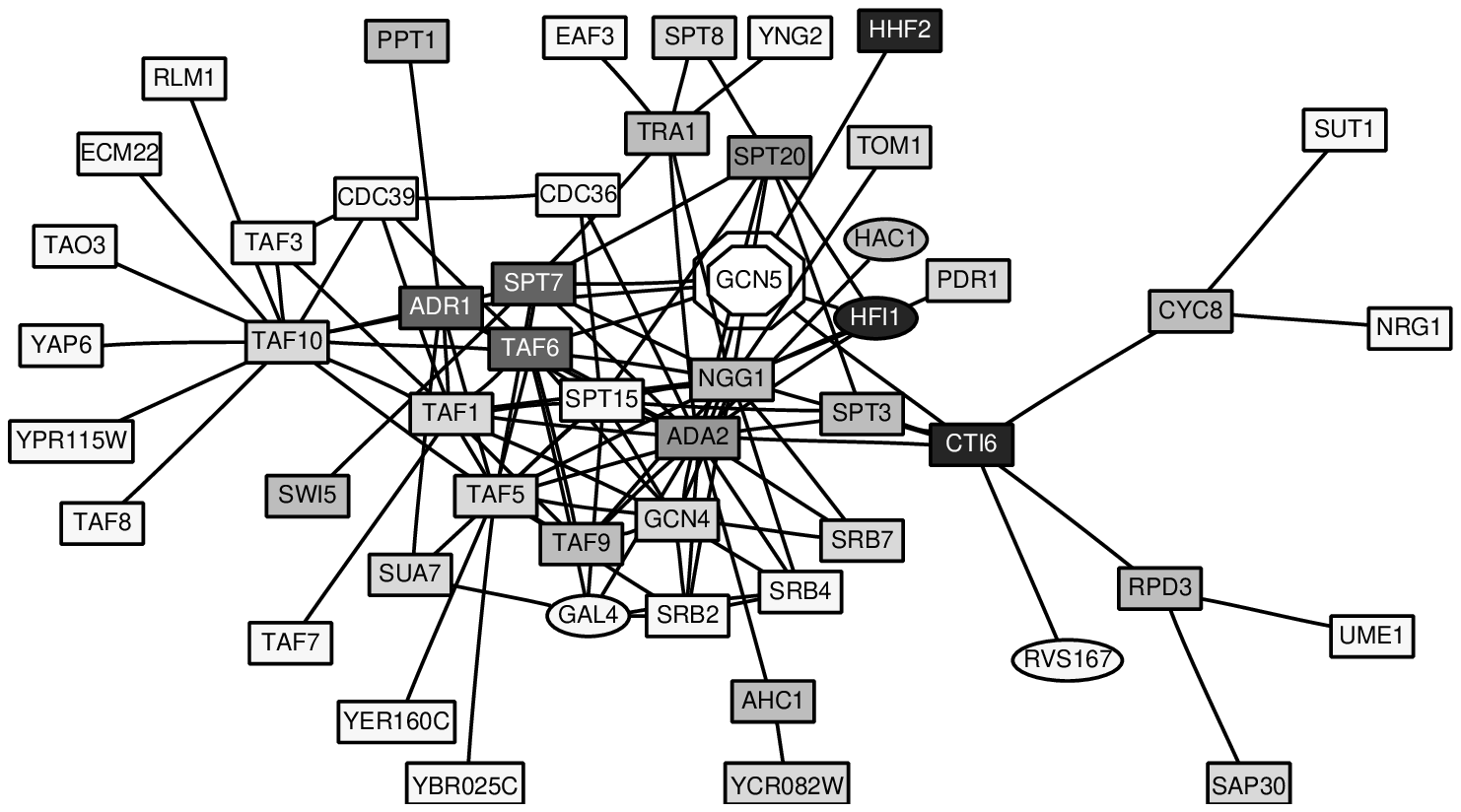} \\
\textbf{(c)} & \\
&\includegraphics[scale=0.70]{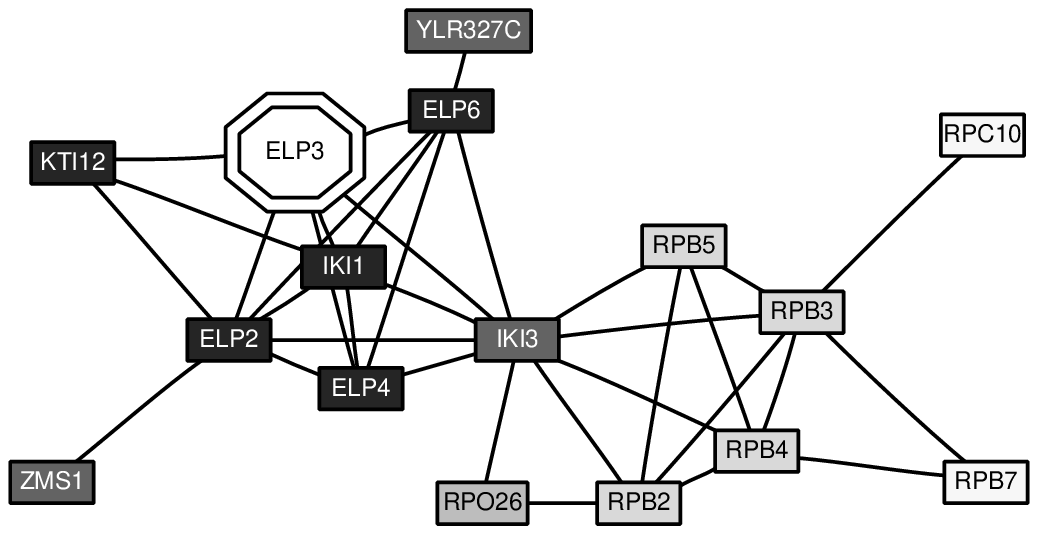} \\
\end{tabular}
\end{center}
\caption[ITMs obtained by running the absorbing model with Esa1(a), Gcn5(b) and Elp3(c) as a sink.]{ITMs obtained by running the absorbing model with Esa1(a), Gcn5(b) and Elp3(c) as a sink. The shades of grey at the nodes represent the probability of the information originating at the corresponding protein being absorbed at the sink, the darker nodes indicating higher probability.}\label{fig:fig1}
\end{figure}

Figure \ref{fig:fig1} shows the three subgraphs of the yeast core interaction graph consisting of the top scoring nodes according to the absorbing model with Esa1, Gcn5 and Elp3 as single sinks, respectively. The information orginating at the proteins shown has more than $0.07$ probability of being absorbed by the sink (under the influence of the potential centered at the sink) as opposed to being dissipated. Hence, the subgraphs show the proteins that are likely to be in the same complex with the HATs chosen as sinks.

Figure \ref{fig:fig1}(a), with Esa1 as the sink, shows all the proteins from the NuA4 complex that are available in the core dataset as highly significant. Some of the proteins from ADA and SAGA complexes can also be seen because Tra1 belongs to these complexes as well as to NuA4. The four types of histones forming the histone octamer can also be seen interacting with Arp4. The proteins Vps51--54 on the right of  Figure \ref{fig:fig1}(a) belong to the Vps Fifty-three thethering (VFT) complex, involved in vesicle assembly \citep{RWSSK03}. The proteins Tlg1 and Ypt6 are interacting partners of the VFT complex \citep{RWSSK03}. The relation between VFT and NuA4 is not established as these two complexes are localized in different cellular compartments: NuA4 in the nucleus and VFT in golgi-vacuole transport vesicles. The relationship observed in Figure \ref{fig:fig1}(a) results exclusively from the Yng2--Vps51 interaction, which was orginally observed in a yeast-two-hybrid screen by \citet{ITMOCNYKS00,ICOYHS01}. Based on the above information, it appears that VFT and NuA4 complexes do not interact \emph{in vivo}. Note that the histones as well as actin, although selected as evaporating points, can be seen in the figure because the outgoing flow from evaporating nodes is allowed. 

In a similar fashion, Figure \ref{fig:fig1}(b), with Gcn5 as the sink, shows the members of SAGA, ADA and TFIID transcriptional activator complexes as well as many other transcription factors, mostly members of subcomplexes of the RNA polymerase II holoenzyme. Also worth mentioning is Cti6, which bridges the Cyc8-Tup1 corepressor and the SAGA coactivator to overcome repression of the GAL1 gene \citep{PPKTT02}. The Cyc8 protein is also shown while Tup1 is not, most likely because it is involved in many other interactions  away from Gcn5, bringing down its relative significance. Figure \ref{fig:fig1}(c), with Elp3 as the sink, clearly outlines the elongator complex, as well as some members of the core RNA polymerase II complex (Rbp2--5, Rbp7, Rpc10, Rpo26) \citep{MY98}.

Figure \ref{fig:fig2} shows the top scoring nodes according to the absorbing model with Esa1, Gcn5 and Elp3 as simultaneous sinks with attracting potentials. In this case, the information originating at the depicted nodes has more than $0.05$ total probability of being absorbed by any of the sinks as opposed to being dissipated. 

Fewer nodes can be seen in this figure as compared to Figure \ref{fig:fig1} because the three attracting potentials are now involved that may cancel each other out. It can be seen that the elongator complex centered around Elp3 is not connected to the subgraph around Esa1 and Gcn5. Although all of the NuA4, SAGA, ADA and elongator complexes belong to the RNA polymerase II holoenzyme, they do so at different times. The NuA4, ADA and SAGA complexes have a role in initiation of transcription while the elongator complex is involved in transcript elongation \citep{Martinez02}. The green (mixture of cyan and yellow) color of Tra1 is indicative of the fact that it is a subunit of both Esa1-containing NuA4 complex and the Gcn5-containing SAGA complex.

\begin{figure}
\begin{center}
\includegraphics[scale=0.7]{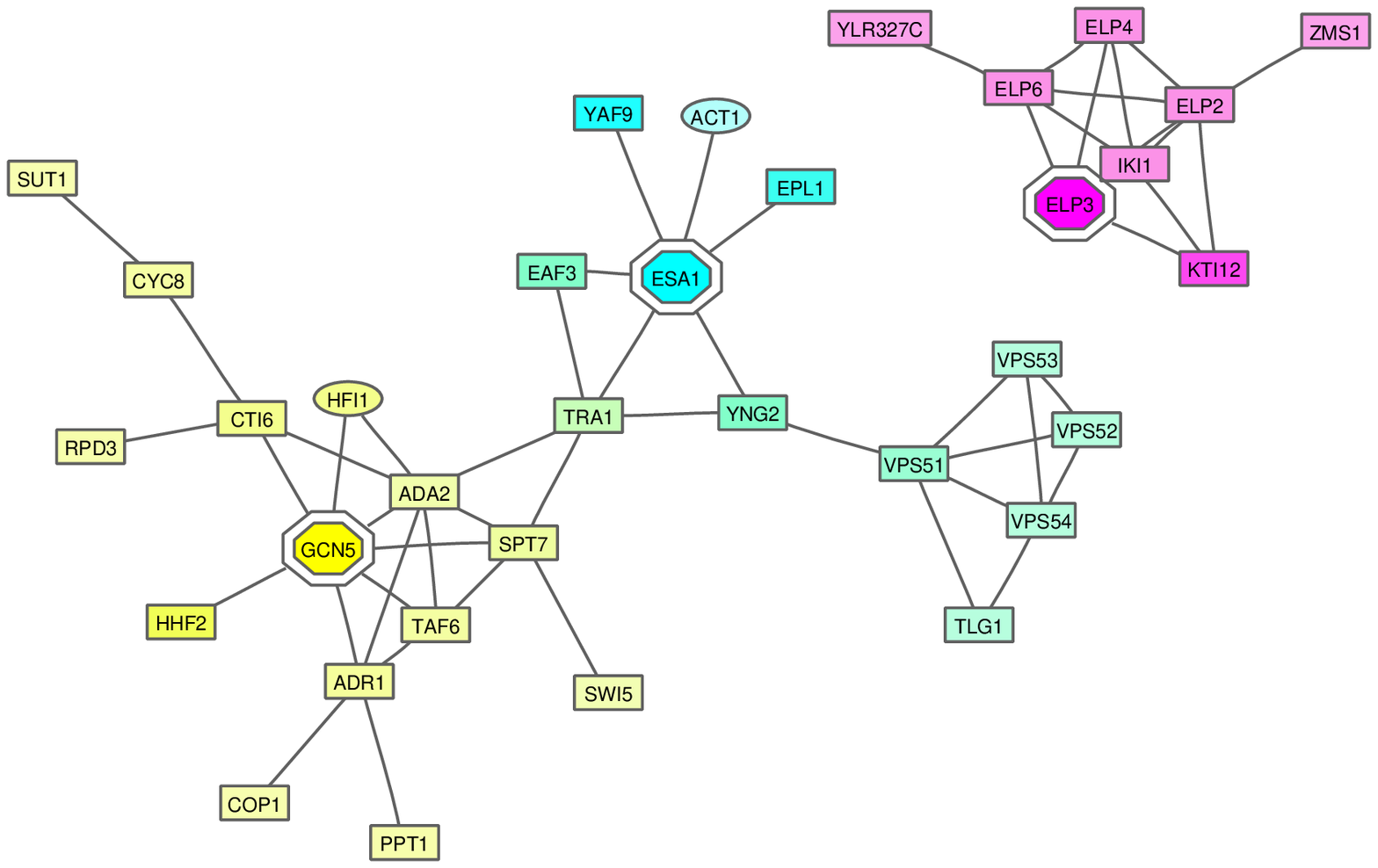}
\end{center}
\caption[ITM obtained by running the absorbing model with Esa1, Gcn5 and Elp3 as simultaneous sinks.]{ITM obtained by running the absorbing model with Esa1, Gcn5 and Elp3 as simultaneous sinks. The strength of each of cyan, yellow and magenta color component of the node shows the square root of the probability of absorption at Esa1, Gcn5 and Elp3, respectively.}\label{fig:fig2}
\end{figure}

\subsection{Transcription factor interaction interfaces; emitting examples}\label{subsec:emitting}

Mcm1 is a yeast transcription factor essential for cell viability. It controls many cellular functions including cell cycle transition \citep{ASWN95}, mating \citep{MBGSAV02} and arginine metabolism \citep{MD93}, through interactions with different cofactors. It has been determined that Mcm1 acts both as an activator and a repressor of transcription \citep{BHS92,MD93} and here we explore the possible ways it can interact with the NuA4 and SAGA HAT complexes.

\begin{figure}
\begin{center}
\begin{tabular}[t]{lr}
\textbf{(a)} & \\
& \includegraphics[scale=0.7]{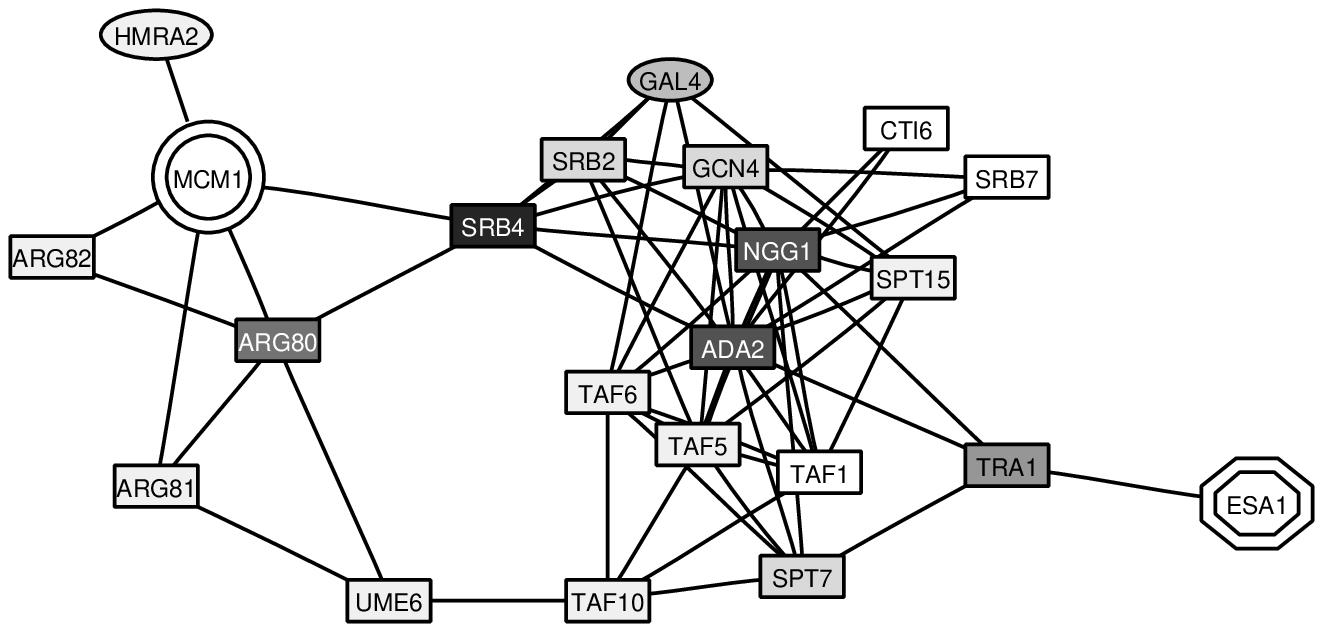} \\
\textbf{(b)} & \\
& \includegraphics[scale=0.7]{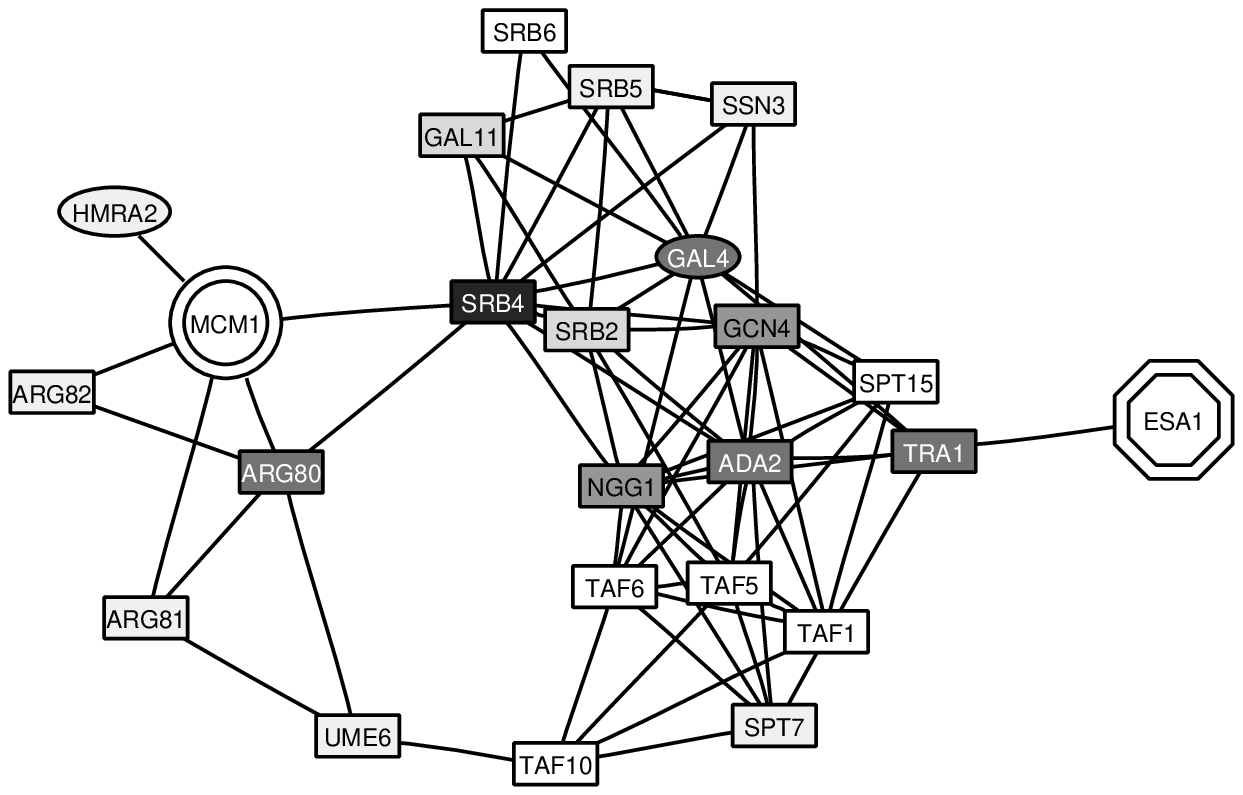} \\
\end{tabular}
\end{center}
\caption[ITMs resulting from the emitting model with Mcm1 as a source and Esa1 as a pseudosink.]{ITMs resulting from the emitting model with Mcm1 as a source and Esa1 as a pseudosink using the original yeast core dataset (a) and the modified dataset additionally including the edges Tra1--Gal4 and Tra1--Gcn4 (b). The proteins containing the largest amounts of deposited information are shown, with the information content indicated by shading (darkest nodes contain the most information).}\label{fig:emitting1}
\end{figure}

Figure \ref{fig:emitting1}(a) shows the subgraph consisting of the $22$ proteins with the largest deposited information content obtained by running our emitting model with Mcm1 as a source and Esa1 as a pseudosink. The number of proteins to display ($20$ plus the source and the pseudosink) was chosen because the participation ratio for the information content vector (excluding the source and the pseudosink) was $20.33$. 

The ITM shown in Figure \ref{fig:emitting1}(a) gives the likely pathways of physical interaction from Mcm1 to Esa1, according to the yeast core interaction dataset. It can be immediately observed that Esa1 is reached solely through Tra1, which is known to be the general interaction domain of both NuA4 and SAGA HAT complexes \citep{AUSCGBPWC99,GSPY98}. Directly associated with Mcm1 are the proteins Arg80--Arg82, belonging to the ArgR complex involved in regulation of arginine metabolism \citep{DM91}. The majority of the ITM is dominated by the members of the SRB mediator subcomplex of the RNA polymerase II holoenzyme (Srb2, Srb4, Srb7) \citep{BY05} and the TFIID, SAGA and ADA complexes. Also prominent are transcriptional activators Gal4 and Gcn4 \citep{Hinnebusch05,TJS06}. 

The subgraph image suggests two possible interaction pathways: the main (based on the intensities of deposited information) through Srb4 and members of SAGA/ADA complex and the alternative through Ume6--TAF10--Spt7. Ume6 is a DNA binding protein that acts as a transcriptional repressor by recruiting histone deacetylases, which have the catalytic activity opposite to the HATs \citep{KABR03}. While  simultaneous existence of activating and repressing pathways is biologically plausible, we do not anticipate both pathways to be in action at the same time. On the other hand, interaction of Mcm1 with the NuA4 through any of the above pathways \emph{in vivo} is doubtful because both pathways lead through the interacting partners of Tra1 in the SAGA complex that are not associated with it in the NuA4 complex \citep{DC04,TT05}. Note that the direct physical interaction of the ArgR/Mcm1 complex and the SAGA complex was hypothesized by \citet{RGB02} in relation to regulation of arginine metabolism.

Nevertheless, it is likely that the yeast core dataset does not contain all the interactions of Tra1 and that the interactions not in the dataset may provide us with the plausible explanation. \citet{BHSA01} have indicated that HAT complexes are recruited through Tra1 by Gal4 and Gcn4 transcriptional activators. To investigate if adding the implied edges would significantly change the resulting ITM we added the Gcn4--Tra1 and Gal4--Tra1 links to the core dataset and rerun the emitting model with all other parameters unchanged. The resulting ITM, with participation ratio of $21.66$, is shown in Figure \ref{fig:emitting1}(b). We observe few changes: the proteins Ssn3, Srb5, Srb6 and Gal11, belonging to the mediator complex, replaced Cti6 and Srb7, thus placing more emphasis to the mediator complex. 

In this example, our emitting model appears to be quite robust to changes in the pseudosink leakage parameter $\gamma$. Using the original core dataset, in addition to the original run with $\gamma=0.3$, we ran our model with $\gamma=0$, $\gamma=0.5$ and $\gamma=1$, obtaining participation ratios of $19.43$, $20.34$ and $20.75$ and very little change in constitution of the ITMs. For example, when $\gamma=1$, the new ITM contains the NuA4 proteins Arp4 and Yng2 in the place of Cti6 and Srb7. Hence, larger pseudosink leakage coefficient allows exploration of the nodes surrounding the pseudosinks without affecting the remainder of the ITM in a major way. Such exploration is very desirable for protein-protein interaction networks because it reveals more of the complexes around pseudosinks, thus giving some of the characteristics of the absorbing model to the emitting model. Note that many of the interacting partners of the sources are found in the ITM solely due to proximity of the source.

\begin{figure}
\begin{center}
\begin{tabular}[t]{lr}
\textbf{(a)} & \\
& \includegraphics[scale=0.7]{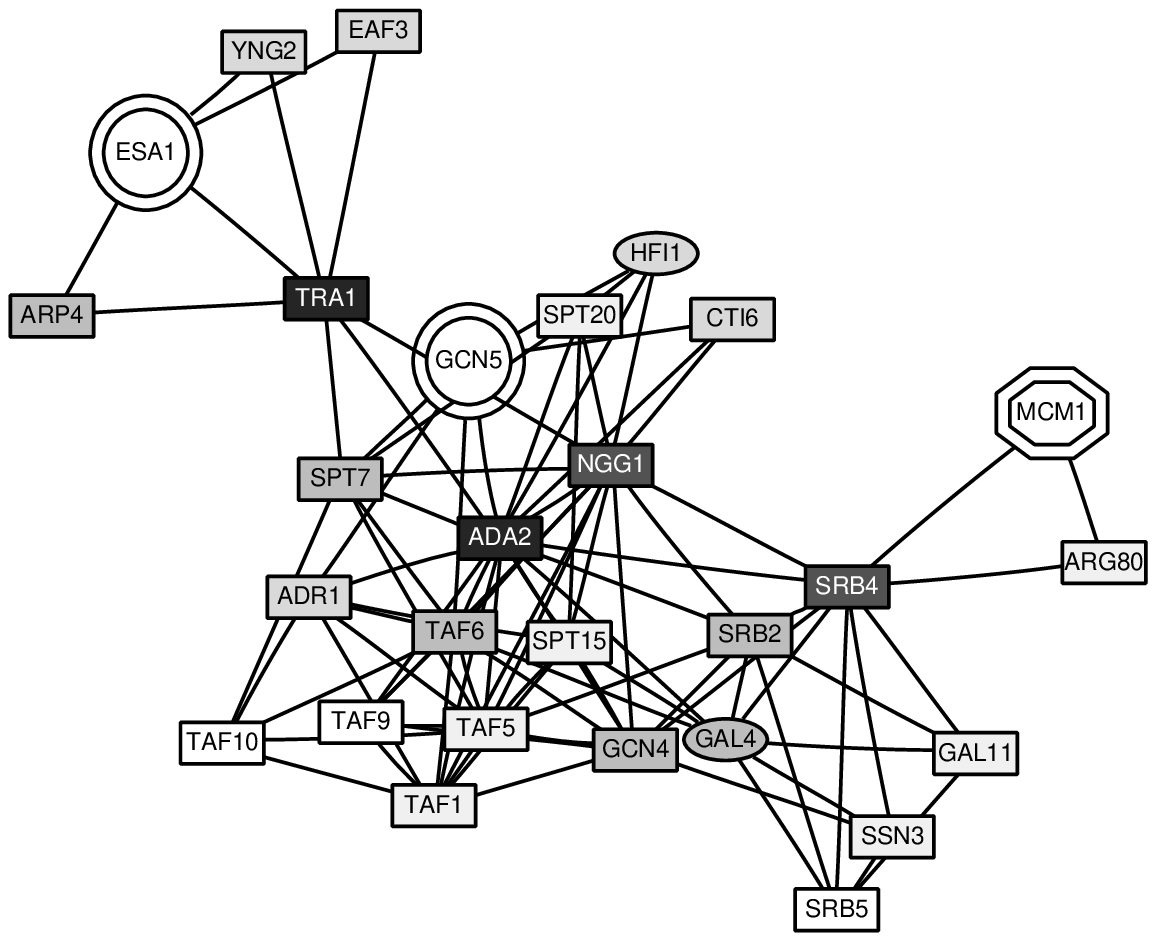} \\
 \textbf{(b)} & \\
& \includegraphics[scale=0.7]{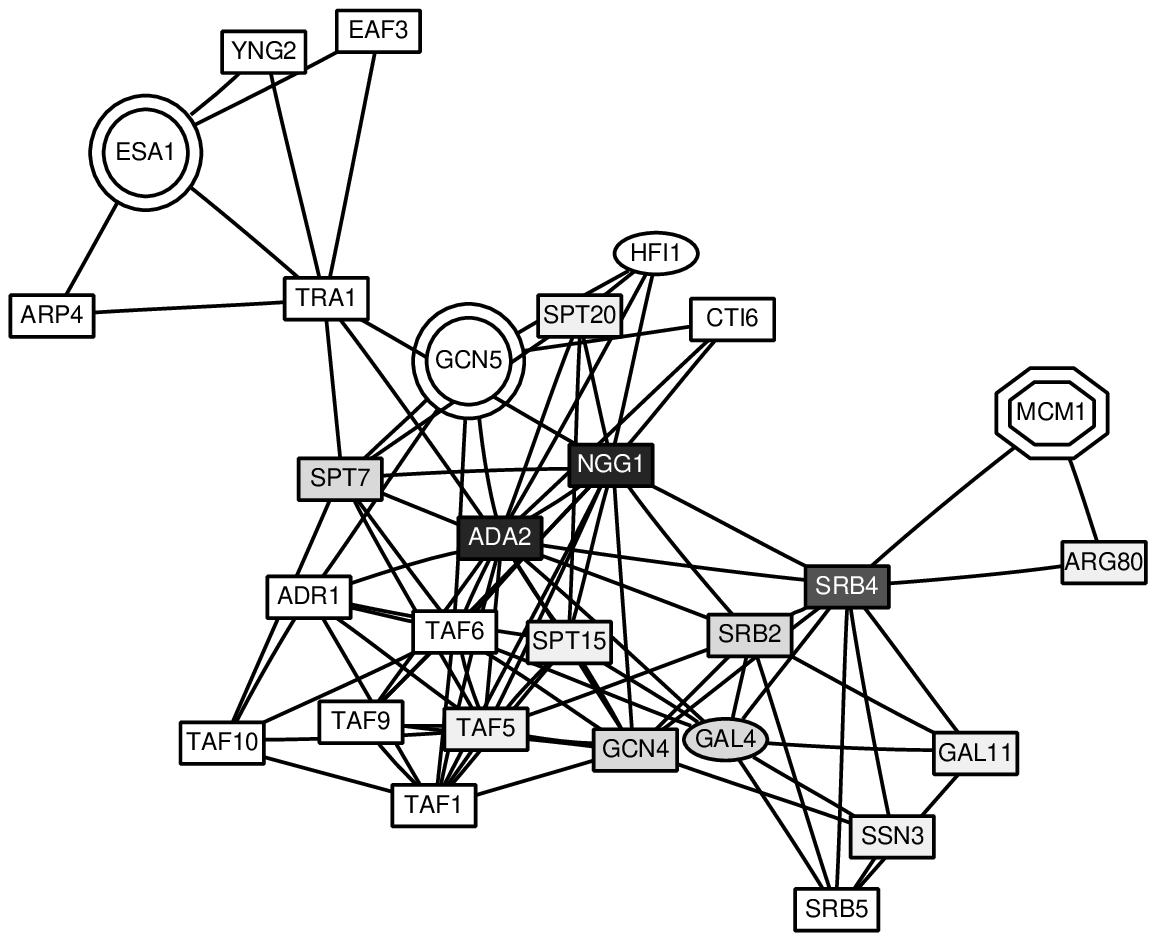} \\
\end{tabular}
\end{center}
\caption[ITM resulting from the emitting model with Esa1 and Gcn5 as sources and Mcm1 as a pseudosink.]{ITM resulting from the emitting model with Esa1 and Gcn5 as sources and Mcm1 as a pseudosink: (a) information content, (b) interference strength.}\label{fig:emitting2a}
\end{figure}

To explore the extend the HATs Esa1 and Gcn5 share their interaction interface with Mcm1 we set Esa1 and Gcn5 as sources and Mcm1 as a pseudosink destination. Figure \ref{fig:emitting2a} shows the ITM based on the total information content (participation ratio $24.62$, 28 nodes shown), with the nodes shaded according to total content and interference strength. The proteins shown as nodes in Figure \ref{fig:emitting2a} have appeared in one of the previous figures, mostly forming parts of NuA4, SAGA/ADA, TFIID and mediator complexes. The nodes with the largest total content are Tra1, Ada2, Ngg1 and Srb4 and the latter three are also the nodes with by far the largest interference strength. This fact does not surprise us because although Tra1 is a member of both NuA4 and SAGA complexes, information flowing from Gcn5 to Mcm1 largely avoids it. 

The paths used by the information emitted from Esa1 and Gcn5 separately can best be seen in a color figure (Figure \ref{fig:emitting2b}(a)) where the information content from Esa1 and Gcn5 is shown as cyan and yellow, respectively. The nodes colored strongly cyan contain mostly information from Esa1 while those colored yellow contain mostly the information from Gcn5. The nodes colored green contain information from both sources. In this way it can be observed that members of NuA4 contain the information solely from Esa1, some SAGA proteins contain the information solely from Gcn5, while Ada2, Ngg1 and Srb4 contain a significant amount of information from both sources.

\begin{figure}
\begin{center}
\begin{tabular}[t]{lr}
\textbf{(a)} & \\
& \includegraphics[scale=0.7]{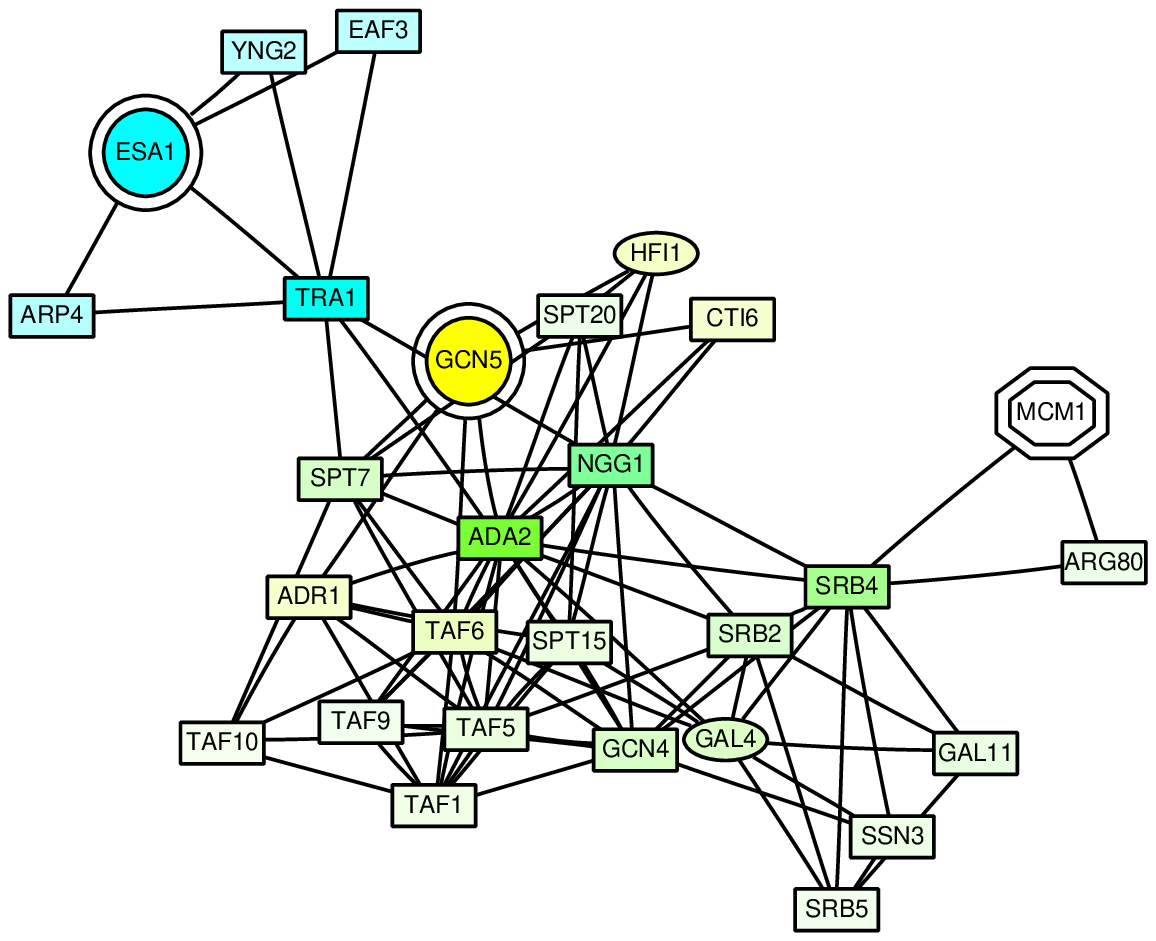} \\
 \textbf{(b)} & \\
& \includegraphics[scale=0.7]{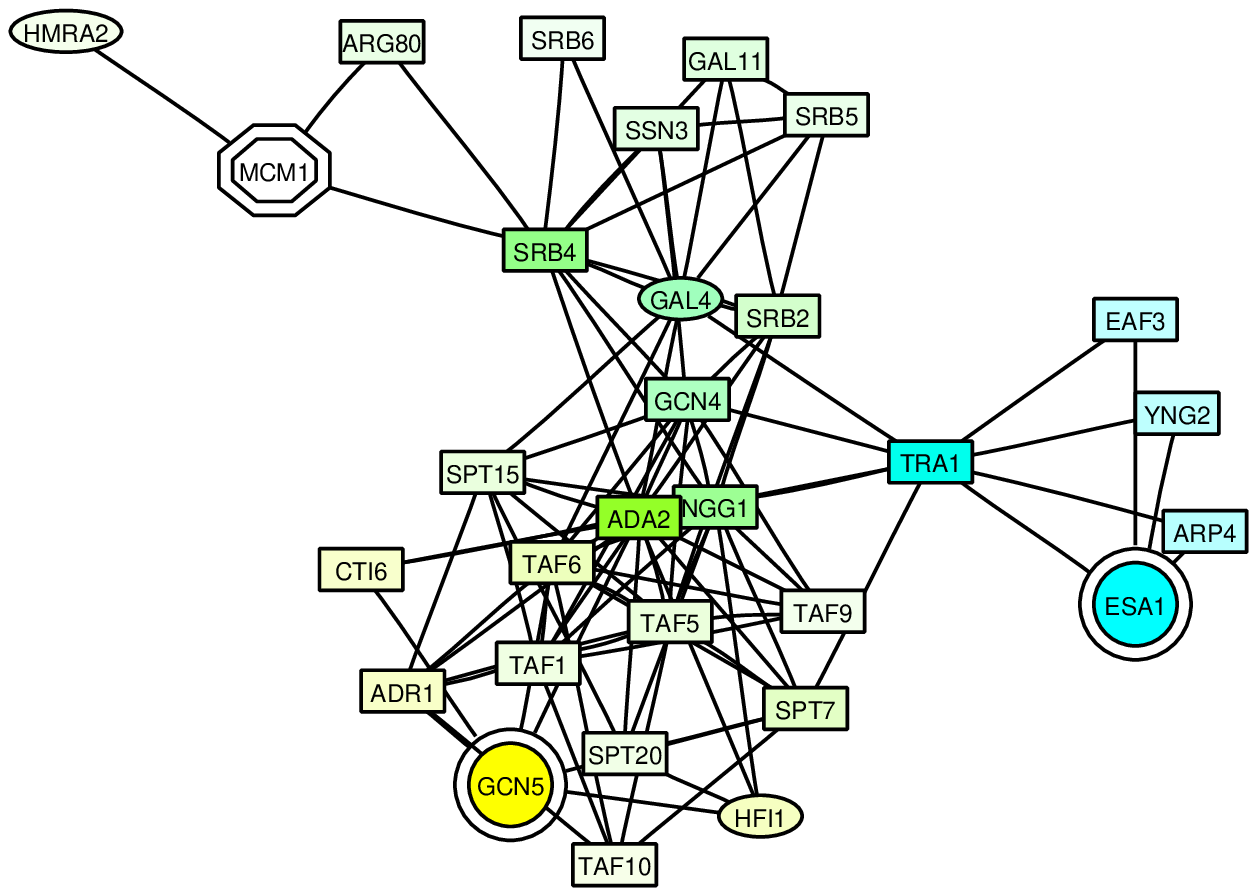} \\
\end{tabular}
\end{center}
\caption[Information content of members of the ITM arising from the emitting model with Esa1 and Gcn5 as sources and Mcm1 as a pseudosink.]{Information content of members of the ITM arising from the emitting model with Esa1 and Gcn5 as sources and Mcm1 as a pseudosink: (a) using the yeast core dataset; (b) using the modified dataset additionally including the edges Tra1--Gal4 and Tra1--Gcn4. The strength of the cyan and yellow color component of the node corresponds to the information content originating from Esa1 and Gcn5, respectively.}\label{fig:emitting2b}
\end{figure}

Using additional links based on Brown \emph{et al.} (Figure \ref{fig:emitting2b}(b)) produces effects similar to Figure \ref{fig:emitting1}(b): the common interface through the mediator complex is emphasized at the expense of the paths through the SAGA complex. For example, note the difference in color of Spt7, Gcn4 and Gal4 between Figure \ref{fig:emitting2b}(a) and Figure \ref{fig:emitting2b}(b). The common interface through the mediator complex appears biologically more plausible than directly through members of the SAGA complex but we are as yet unable to find direct evidence in the literature confirming either possibility.

\section{Discussion and conclusion}

The proposed information diffusion models appear to capture some of the essential features of the yeast protein-protein interaction network in our examples. Our absorbing model performed well in identifying complexes related to sinks while the emitting model with pseudosinks is able to illuminate the possible interaction interfaces between sources and pseudosinks. Application of the concept of destructive interference in this context provides a way to assess the degree of overlap of different ITMs.

The salient feature of our models is a novel use of attraction potentials and dissipation. While the entries of the Green's function can be interpreted in graph-theoretic terms as sums of weights of paths from a source to a transient vertex (for the emitting model) or from a transient vertex to a sink (for the absorbing model), the potentials, together with the choice of boundary, provide a unique context for information diffusion in the network. The weights of the edges and hence the nature of the underlying graphs are changed every time a different potential is applied, thus bringing forward different aspects of the network. The potential function used for our examples was heuristic in nature and we hope that our work would generate interest in developing theoretical foundations for directed information propagation through networks.

Dissipation coefficients provide a natural and extremely flexible way of controlling the spread of information content through the network. While \citet{GN02} proposed a similar formulation for penalizing longer paths connecting two nodes in a network, they did so in the context of hierarchical clustering and using a single dissipation rate. Node specific dissipation rates are important because they allow construction of `evaporating nodes' and possible integration of additional information to our model. Having the dissipation rates dependent on the environment of the node may lead to a more sophisticated model of information transduction.

When modelling physical cellular protein networks, the main limitation of our approach is that the the publicly available representations of protein-protein interaction networks contain a limited amount of information. Each interaction is shown as either occurring or not occurring, without reference to the dynamics, time-scale, or specificity of binding. Furthermore, the spatial location of the interactions on the protein molecules is not available, so that it cannot be determined if a protein known to belong to two separate complexes, such as Tra1 in our examples, can belong to both at the same time and therefore transmit information between them. Therefore, our model of protein cellular networks is only metaphorical at this stage. However, our diffusion paradigm can be adapted to account for additional information about proteins, such as their concentrations, cellular compartment localizations, post-translational modifications or rate constants for binding interactions, as it becomes available. One way to do that is to associate each protein to a vector instead of a scalar value and to construct an evolution operator that reflects the nature of the additional information. In such circusmstances, the dynamics of information flow could be as revealing as the steady state we use at this stage. 

The quality of the interaction dataset also has a strong influence to the outcomes of our models. Addition or deletion of edges may make the results more realistic, as in our emitting examples, but also may completely alter the ITM produced, if a particular edge provides a shortcut towards the destination. Hence, in order to obtain the results useful in field of application, it is imperative to use datasets of interactions that precisely reflect the network being investigated. In the case of yeast protein-protein interactions, \citet{CKZG07} were recently able to derive a significantly more reliable collection of interactions, primarily based on two large-scale studies of protein complexes by tandem affinity purification of complexes followed by mass spectroscopic identification of individual proteins \citep{GAGK06,KCYZ06}. It is interesting that the same transcriptional complexes encountered in our examples are prominent in the unified physical interactome map presented by \citet{CKZG07}.

The problem of `shortcuts' through the network was also observed by \citet{SPAD02}, who completely eliminated certain nodes in their effort to model signal transduction pathways using the yeast protein-protein interactions. Our evaporating nodes, with a very large incoming dissipation rate, have a similar role with an added advantage that they can be visible as parts of complexes observed using the absorbing model. The list of evaporating nodes used by us is not exhaustive and it would be necessary to add further classes of proteins to it for large-scale investigations of the yeast protein interactome using our methods.

In this paper, we introduced a flexible mathematical framework for analysis of interaction networks and indicated its utility by examples. We believe that the ability to select a particular context for information propagation by setting various model parameters will be extremely useful for addressing questions involving interaction networks in biology and many other disciplines.

\section{Acknowledgments}

We thank Drs John Wootton and David Landsman for encouragement and comments. This work was supported by the Intramural Research Program of the National Library of Medicine at National Institutes of Health.

\newpage
\appendix
\section{Existence of Green's Functions}\label{app:Green}

In this appendix we provide the elementary proofs of the results about existence of the Green's functions stated in the main text. As before, $\Gamma=(V,E,w)$ denotes a weighted directed graph with $N$ vertices, with the weight matrix $\mat{W}$ and transition matrix $\mat{P}$. We also have $T\subset V$ and $S=V\setminus T$.

Recall that for every matrix $\mat{M}$, the induced $\ell_\infty$ norm of $\mat{M}$, written $\norm{\mat{M}}_\infty$, is defined by 
\begin{equation}
\norm{\mat{M}}_\infty=\sup_{\vec{x}\in\R^n}\frac{\norm{\mat{M}\vec{x}}_\infty}{\norm{\vec{x}}_\infty},
\end{equation}
where $\norm{\vec{x}}_\infty=\max_i \abs{x_i}$. One can easily show that
\begin{equation}\label{eqn:l1norm}
\norm{\mat{M}}_\infty = \max_i \sum_j \abs{M_{ij}}.
\end{equation}
Also recall that the spectral radius of a square matrix $\mat{M}$ is defined to be the largest absolute value of its eigenvalues. It is well known that that for every eigenvalue $\lambda$ of $\mat{M}$ and any $k=1,2,\ldots$,
\begin{equation}\label{eq:specrad}
\abs{\lambda} \leq \norm{\mat{M}^k}_\infty^{1/k}.
\end{equation}

\begin{lemma}\label{lemma:mat}
Let $\mat{M}$ be a square matrix with the spectral radius strictly less than $1$. Then, 
\begin{enumerate}[(i)]
\item $\mat{M}^k\to \mat{0}$ as $k\to\infty$,
\item The matrix $\I-\mat{M}$ is invertible and $(\I-\mat{M})^{-1} = \sum_{k=0}^\infty \mat{M}^k$.
\end{enumerate}
\end{lemma}
\begin{proof}
By the Jordan matrix decomposition, we can write $\mat{M}=\mat{V}\boldsymbol{\Lambda}\mat{V}^{-1}$ for some matrix $\mat{V}$, where $\boldsymbol{\Lambda}$ is a block-diagonal matrix of the form 
\[\boldsymbol{\Lambda}=\left[\begin{array}{cccc}
\mat{B}_1 & \mat{0} & \cdots & \mat{0}\\
\mat{0} & \mat{B}_2 & \cdots & \mat{0}\\
\vdots & \vdots & \ddots & \vdots\\
\mat{0} & \mat{0} & \cdots & \mat{B}_N 
\end{array}\right],\]
with each of the sub-blocks $\mat{B}_j$,  $1\leq j\leq N$, is of the form $\mat{B}_j =\lambda_j\I + \mat{C}_j$ where
\[\mat{C}_j= \left[\begin{array}{ccccc}
0 & 1 & 0 & \cdots & 0\\
0 & 0 & 1 & \cdots & 0\\
\vdots & \vdots & \ddots & \ddots & \vdots\\
0 & 0 & 0 & \cdots & 1 \\
0 & 0 & 0 & \cdots & 0 
\end{array}\right] 
\]
and $\lambda_1, \ldots \lambda_N$ are eigenvalues of $\mat{M}$. Hence,  $\mat{M}^k=\mat{V}\boldsymbol{\Lambda}^k\mat{V}^{-1}$ and 
\[\boldsymbol{\Lambda}^k=\left[\begin{array}{cccc}
\mat{B}_1^k & \mat{0} & \cdots & \mat{0}\\
\mat{0} & \mat{B}_2^k & \cdots & \mat{0}\\
\vdots & \vdots & \ddots & \vdots\\
\mat{0} & \mat{0} & \cdots & \mat{B}_N^k 
\end{array}\right].\]
For each eigenvalue $\lambda_j$ and each block $\mat{B}_j$, we can write \[\mat{B}_j^k = (\lambda_j\I+\mat{C}_j)^k = \sum_{p=0}^k \binom{k}{p} \lambda_j^{k-p} \mat{C}_j^p.\] It can easily be shown that for each $j$, $\mat{C}_j$ is a nilpotent matrix, that is, if $\mat{C}_j$ is an $m\times m$ matrix, then $\mat{C}^m=\mat{0}$. Therefore, for $k\geq m-1$, \[ \mat{B}_j^k = \lambda_j^{k-m+1} \left(    \sum_{p=0}^{m-1} \binom{k}{p} \lambda_j^{m-p-1} \mat{C}_j^p \right).  \] Observe that the above expression in parenthesis gives an (upper triangular) matrix whose entries are $m-1$-th degree polynomials in $k$ and hence, that the whole expression for $\mat{B}_j^k$ is dominated by $\lambda_j^{k-m+1}$. Since, by the spectral radius assumption, $\abs{\lambda_j}<1$ for each $i$, it follows that for each $j$, $\mat{B}_j^k\to\mat{0}$ as $k\to\infty$ and hence  $\boldsymbol{\Lambda}^k\to\mat{0}$ as $k\to\infty$ by the block structure. This proves the first statement.

For the second statement suppose  that $\I-\mat{M}$ is singular. Then $\I-\mat{M}$ has $0$ as an eigenvalue and hence $\lambda=1$ is an eigenvalue of $\mat{M}$, contradicting our assumption about the spectral radius of $\mat{M}$. Therefore, $\I-\mat{M}$ is invertible. Furthermore, it can easily be obtained using the block diagonal structure of $\boldsymbol{\Lambda}$ and the ratio test that the sum $\sum_{k=0}^\infty \mat{M}^k$ converges, Hence, \[(\I-\mat{M})\sum_{k=0}^\infty \mat{M}^k =  \sum_{k=0}^\infty \mat{M}^k -  \sum_{k=0}^\infty \mat{M}^{k+1}  = \I + \sum_{k=1}^\infty \mat{M}^k -  \sum_{k=1}^\infty \mat{M}^k = \I.\]
\end{proof}

Since the matrix $\mat{P}$ is stochastic, we have $\norm{\mat{P}}_\infty=1$ and hence the spectral radius of $\mat{P}$ is bounded by $1$. Since $\mat{P}_{TT}$ is a submatrix of $\mat{P}$, we have $\norm{\mat{P}_{TT}}_\infty\leq 1$ and its spectral radius is also bounded by $1$. To prove Proposition \ref{prop:specrad0} (denoted Proposition \ref{prop:specrad} below) we will show that the spectral radius of $\mat{P}_{TT}$ is strictly smaller than $1$ if there is some vertex in $S$ that can be reached from any transient node via a directed path. Before presenting the main proof, we require several lemmas.

\begin{lemma}\label{lemma:sprad1}
Let $\mat{B}$ and $\mat{C}$ be $n\times n$ matrices with non-negative entries such that $\norm{\mat{B}}_\infty\leq 1$ and  $\norm{\mat{C}}_\infty\leq 1$ and let $\mat{D}=\mat{C}\mat{B}$. Suppose there exists $1\leq p\leq n$ such that $0 < \sum_j B_{pj}< 1$. Then, for every $1\leq i\leq n$ such that $C_{ip} > 0$, \[\sum_j D_{ij} < 1.  \]
\end{lemma}
\begin{proof}
Let $K=\{k: C_{ik}>0 \}$. Then $p\in K$ and
\begin{align*}
\sum_j D_{ij} & = \sum_j \sum_k C_{ik}B_{kj}\\
& =   \sum_{k\in K} C_{ik} \sum_j B_{kj}\\
& \leq \sum_{k\in K\setminus\{p\}} C_{ik} \norm{\mat{B}}_\infty + C_{ip} \sum_j B_{pj}\\
& < \sum_{k\in K\setminus\{p\}} C_{ik}+ C_{ip}\\
& \leq 1.
\end{align*}
\end{proof}

\begin{lemma}\label{prop:specrad3}
Let $\Gamma$ be a weighted directed graph with weight matrix $\mat{W}$. Let $i$ and $j$ be distinct nodes of $\Gamma$ connected by a directed path from $i$ to $j$ of length $n\geq 1$. Then $W^n_{ij}>0$. 
\end{lemma}
\begin{proof}
We use induction. If $i$ and $j$ are connected with a path of length $1$, then there exists an edge $(i,j)\in E$ and hence $W_{ij}>0$. Assume that $W^m_{ij}>0$ if $i$ and $j$ are connected by a directed path from $i$ to $j$ of length $m$. Suppose $i$ and $j$ are connected by a path of length $m+1$. Then there exists a vertex $k$ such that $i$ and $k$ 
are connected by a directed path from $i$ to $k$ of length $m$ and there exists a directed edge $(k,j)$. Hence, by our assumption $W^m_{ik}>0$ and $W_{kj}>0$. Therefore,
\[W^{m+1}_{ij} = \sum_{k'\in V} W^{m}_{ik'}W_{k'j}\geq W^m_{ik}W_{kj} > 0. \]
\end{proof}

\begin{lemma}\label{lemma:sprad4}
Let $\mat{M}=\mat{P}_{TT}$, let $i\in T$ and suppose there exists $s\in S$ such that there exists a directed path from $i$ to $s$ of length $m$. Then for all $n\geq m$, 
\begin{equation}
\sum_{k\in T} M^{n}_{ik} < 1.
\end{equation}
\end{lemma}
\begin{proof}
Let $i\in T$ and let $s\in S$ be a vertex such that there exists a directed path from $i$ to $s$ of length $m$. Let $J$ be the set of vertices in $T$ directly adjacent to a vertex in $S$. Then, by our assumption, for every $i\in T$ there exists $j\in J$ such that there exists a directed path from $i$ to $j$ of length $m-1$. Since the matrix $\mat{P}_{TT}$ can be treated as the weight matrix for the subgraph of $\Gamma$ restricted to vertices in $T$, it follows by Lemma \ref{prop:specrad3} that $M^{m-1}_{ij}>0$.

Since every point in $J$ is adjacent to a point in $S$, it also follows that $\sum_{k\in T} M_{jk}< 1$. Clearly, $\norm{\mat{M}}_\infty\leq 1$ and hence $\norm{\mat{M}^{m-1}}_\infty\leq 1$. Applying Lemma \ref{lemma:sprad1} to the matrices $\mat{M}$ and $\mat{M}^{m-1}$ we obtain that for every $i\in T$, $\sum_{k\in T} M^{m}_{ik} < 1$. 

Let $t\geq m$ and assume $\sum_{k\in T} M^{t}_{ik} < 1$. We have 
\[\sum_{k\in T} M^{t+1}_{ik} = \sum_{k\in T} \sum_{k'\in T} M^t_{ik'}M_{k'k} =   \sum_{k'\in T} M^t_{ik'} \sum_{k\in T} M_{k'k} \leq \sum_{k'\in T} M^t_{ik'} < 1 \]
and our result follows by induction.
\end{proof}

\begin{prop}\label{prop:specrad}
Suppose that for every $p\in T$ there exists $s\in S$ such that there exists a directed path from $p$ to $s$. Then, the matrix $\I-\mat{P}_{TT}$ is invertible and 
\begin{equation}\label{eq:invassum}
(\I-\mat{P}_{TT})^{-1} = \sum_{k=0}^\infty (\mat{P}_{TT})^k.
\end{equation}
\end{prop}
\begin{proof}
Let $\mat{M}=\mat{P}_{TT}$. Observe that our assumption implies that for every $i\in T$ there exists $s\in S$ such that there exists a directed path from $i$ to $s$ of length at most $N$. By Lemma \ref{lemma:sprad4}, we have for every $i\in T$, $\sum_{k\in T} M^{N}_{ik} < 1$. Hence, $\norm{\mat{M}^{N}}_\infty< 1$ and therefore the spectral radius of $\mat{M}=\mat{P}_{TT}$ is strictly smaller than $1$. Our result follows by Lemma \ref{lemma:mat}.
\end{proof}

\subsection{Information dissipation}\label{app:dissipation}

\begin{prop}\label{prop:dissipation1}
Let $\balpha$ and $\bbeta$ be vectors of length $N$ such that for all $i\in V$, $\alpha_i>0$ and $\beta_i>0$. Define the $N \times N$ matrix $\mat{\tilde{P}}$ with entries
\begin{equation*}
\tilde{P}_{ij} = \alpha_i\beta_j P_{ij},
\end{equation*}
Let $\alpha_*=\max\{\alpha_i:i\in V\}$ and $\beta_*=\max\{\beta_i:i\in V\}$ and suppose $\alpha_*\beta_*<1$. Then, the matrix $\I-\mat{\tilde{P}}_{TT}$ is invertible and 
\begin{equation}\label{eq:invassum2}
(\I-\mat{\tilde{P}}_{TT})^{-1} = \sum_{k=0}^\infty (\mat{\tilde{P}}_{TT})^k.
\end{equation}
\end{prop}
\begin{proof}
Let $\mat{M}=\mat{\tilde{P}}_{TT}$ and let $i\in T$. Then, 
\begin{align*}
\sum_{j\in T} M_{ij} & = \sum_{j\in T} \alpha_i\beta_jP_{ij} \leq \alpha_*\beta_* \sum_{j\in T} P_{ij}< 1.
\end{align*}
Hence, $\norm{\mat{M}}_\infty<1$ and thus the spectral radius of $\mat{\tilde{P}}_{TT}$ is strictly smaller than $1$. Our result then follows by Lemma \ref{lemma:mat}. 
\end{proof}

More generally, it is possible to interpret dissipation in the light of Proposition \ref{prop:specrad} by constructing a new graph $\tilde{\Gamma}$ with the vertex set $\tilde{V}=V\cup\{v\}$, where $v$ denotes an additional vertex. The weight matrix of $\tilde{\Gamma}$, denoted $\mat{\tilde{W}}$, has entries
\begin{equation}
\tilde{W}_{ij} = \begin{cases}
\alpha_i\beta_j P_{ij} & \text{if $i\in V$ and $j\in V$,}\\
1 - \sum_{k\in V} \alpha_i\beta_k P_{ik} & \text{if $i\in V$ and $j=v$,}\\
0 & \text{if $i=v$.}
\end{cases}
\end{equation}
Clearly, a random walk on $\tilde{\Gamma}$ is equivalent to a random walk on $\Gamma$ with dissipation: the dissipated information is directed towards the additional vertex $v$ and then disappears. If we place $v$ in the boundary set $\tilde{S}$, by Proposition \ref{prop:specrad}, the necessary condition for existence of the Green's function $(\I-\mat{\tilde{P}}_{TT})^{-1}$ is that from every transient node $i$ there exists a directed path to either a node $s\in S$ or a node $j\in T$ such that $\sum_{k\in V} \alpha_j\beta_k P_{jk}< 1$ (such node $j$ is adjacent to $v$ in the graph $\tilde{\Gamma}$. Proposition \ref{prop:dissipation1} then just represents the special case where every transient vertex is adjacent to $v$ in $\tilde{\Gamma}$.

\section{Interpretations of the matrices $\mat{F}$ and $\mat{H}$ }\label{app:interpr}

\subsection{$\mat{F}$ and $\mat{H}$ as matrices of expected visiting times} \label{app:interpr1}

We will show that both $F_{ij}$ and $H_{ij}$ can be interpreted as the expected number of times a random walk originating at the vertex $i$ visits the vertex $j$, while avoiding all vertices in the boundary set $S$. Note that in the case of the matrix $\mat{F}$, we have $i\in T$ and $j\in S$ while for the matrix $\mat{H}$, $i\in S$ and $j\in T$. We will use $\E$ to denote the expectation operator.

\begin{lemma}
Suppose the boundary set $S$ represents sinks and let $Z_{ij}$ be a random variable denoting the total number of times a random walk starting at $i\in T$ is absorbed at $j\in S$. Then,
\begin{equation}
\E(Z_{ij}) = F_{ij}.
\end{equation}
\end{lemma}
\begin{proof}
Let $Y_{ij}(t)$ be the random variable taking the value $1$ if the random walk originating at $i\in T$ is absorbed at $j\in S$ at time $t$, with probability $\sum_{k\in T} P_{ik}^{t-1} P_{kj}$, and taking the value $0$ otherwise. We have $Z_{ij}=\sum_{t=1}^\infty Y_{ij}(t)$ and $\E(Y_{ij}(t)) = \sum_{k\in T} P_{ik}^{t-1} P_{kj}$. Thus,
\begin{align*}
\E(Z_{ij}) & = \E\left(\sum_{t=1}^\infty Y_{ij}(t)\right)\\
& = \sum_{t=1}^\infty \E(Y_{ij}(t))\\
& = \sum_{t=1}^\infty \sum_{k\in T} P_{ik}^{t-1} P_{kj}\\
& = \sum_{k\in T} \sum_{t=0}^\infty P_{ik}^{t}P_{kj}\\
& = \sum_{k\in T} G_{ik} P_{kj}\\
& = F_{ij}. \qedhere
\end{align*}
\end{proof}

\begin{lemma}
Suppose the boundary set $S$ represents sources and let $Z_{ij}$ be a random variable denoting the total number of times a random walk starting at $i\in S$ visits the node $j\in T$. Then,
\begin{equation}
\E(Z_{ij}) = H_{ij}.
\end{equation}
\end{lemma}
\begin{proof}
In the same fashion as above, let $Y_{ij}(t)$ be the random variable taking the value $1$ if the random walk originating at $i\in S$ is at $j\in T$ at time $t$, with probability $\sum_{k\in T} P_{ik} P^{t-1}_{kj}$, and taking the value $0$ otherwise. We have $Z_{ij}=\sum_{t=1}^\infty Y_{ij}(t)$ and $\E(Y_{ij}(t)) = \sum_{k\in T} P_{ik} P_{kj}^{t-1}$. Thus,
\begin{align*}
\E(Z_{ij}) & = \E\left(\sum_{t=1}^\infty Y_{ij}(t)\right)\\
& = \sum_{t=1}^\infty \E(Y_{ij}(t))\\
& = \sum_{t=1}^\infty \sum_{k\in T} P_{ik} P^{t-1}_{kj}\\
& = \sum_{k\in T} \sum_{t=0}^\infty P_{ik}P^{t}_{kj}\\
& = \sum_{k\in T} P_{ik}G_{kj}\\
& = H_{ij}. \qedhere
\end{align*}
\end{proof}

\subsection{Invariants of $\mat{F}$ and $\mat{H}$}\label{app:interpr2}

Let $\vec{1}\in\R^n$ denote the vector whose entries are all $1$'s. Since all rows of $\mat{P}$ sum to unity, it follows that $\mat{P}\vec{1} = \vec{1}$ and hence $\vec{1}$ is a right eigenvector of $\mat{P}$ for the eigenvalue $\lambda=1$. Define $\vec{d}$ as a vector of length $n$ having entries $d_i = \sum_j W_{ij}$. If $\Gamma$ is unweighted graph, $d_i$ gives the degree of the node $i$. Assuming $\mat{W}$ is symmetric,  \[ \sum_k P_{kj}d_k = \sum_k W_{kj} = \sum_k W_{jk} = d_j \] and therefore $\vec{d}$ is a left eigenvector of $\mat{P}$ corresponding to the eigenvalue $\lambda=1$.  This leads to the following result.

\begin{lemma}\label{lemma:Fstoch}
Suppose that the matrix $\I-\mat{P}_{TT}$ is invertible. Let $\vec{u}$ and $\vec{v}$ be 
the left and right eigenvector of the matrix $\mat{P}$ corresponding to the eigenvalue $\lambda=1$, respectively. Write $\vec{u} = [\vec{u}_S\ \vec{u}_T ]$ and $\vec{v}=\left[\begin{array}{c}\vec{v}_S\\ \vec{v}_T \end{array}\right]$. Then, 
\begin{equation}\label{eqn:leigen0}
\vec{u}_T = \vec{u}_S\mat{H}, 
\end{equation}
and
\begin{equation}\label{eqn:reigen0}
\vec{v}_T = \mat{F}\vec{v}_S.
\end{equation}
\end{lemma}
\begin{proof}
Using the canonical form of the matrix $\mat{P}$ (Equation (\ref{eqn:Pcannonical})) and the fact that $\vec{u}$ and $\vec{v}$ are left and right eigenvectors of $\mat{P}$ respectively, we obtain
\begin{equation}\label{eqn:leigen1}
\vec{u}_T = \vec{u}_S\mat{P}_{ST} + \vec{u}_T\mat{P}_{TT}, 
\end{equation}
and
\begin{equation}\label{eqn:reigen1}
\vec{v}_T = \mat{P}_{TS}\vec{v}_S + \mat{P}_{TT}\vec{v}_T.
\end{equation}
Rearranging Equations (\ref{eqn:leigen1}) and (\ref{eqn:reigen1}) leads to 
\begin{equation}\label{eqn:leigen2}
\vec{u}_T(\I-\mat{P}_{TT}) = \vec{u}_S\mat{P}_{ST}, 
\end{equation}
and
\begin{equation}\label{eqn:reigen2}
(\I-\mat{P}_{TT})\vec{v}_T = \mat{P}_{TS}\vec{v}_S.
\end{equation}
Our result then follows as the consequence of invertibility of $\I- \mat{P}_{TT}$.
\end{proof}

Since $\vec{1}$ is a right eigenvector of $\mat{P}$, it follows from (\ref{eqn:reigen0}) that for all $i$, $\sum_{j\in S} F_{ij} = 1$. Furthemore, recall that if $\Gamma$ is an undirected graph, $\mat{W}$ is symmetric and $\vec{d}$ is a left eigenvector of $\mat{P}$ for $\lambda=1$. Assuming the matrix $\mat{H}$ exists, we obtain from Lemma \ref{lemma:Fstoch} that, if $S$ contains a single point, the matrix $\mat{H}$ is a row vector, which is a multiple of $\vec{d}_T$. 

\bibliographystyle{jcbnatshrt} 
\bibliography{networks}

\end{document}